\def\draft{1}
\def\anon{0}

\documentclass[11pt]{article}

\ifnum\anon=1 \def\draft{0} \fi 

\usepackage{amsfonts}
\usepackage{amsmath}
\usepackage{amssymb}
\usepackage{amsthm}
\usepackage{mathabx}
\usepackage{thmtools, thm-restate}
\usepackage{bbm}
\usepackage{bm}
\usepackage{dsfont}
\usepackage{fullpage}
\usepackage{tikz}
\usepackage[most]{tcolorbox}
\usepackage{xcolor}
\usepackage{ytableau}

\newcommand*{\myfont}{\fontfamily{bch}\selectfont}
\DeclareTextFontCommand{\textmyfont}{\myfont}

\usepackage[linesnumbered,ruled,vlined]{algorithm2e} % algorithm package
\SetKwComment{Comment}{/* }{ */}

\SetCommentSty{mycommfont}

\usepackage[hidelinks,hypertexnames=false]{hyperref}
\usepackage[nameinlink]{cleveref}
\usepackage{url}            % simple URL typesetting
\usepackage{footnotebackref}

\newtheorem{theorem}{Theorem}[section]
\newtheorem{corollary}[theorem]{Corollary}
\newtheorem{lemma}[theorem]{Lemma}
\newtheorem{observation}[theorem]{Observation}

\newtheorem{definition}[theorem]{Definition}
\newtheorem{claim}[theorem]{Claim}
\newtheorem{fact}[theorem]{Fact}
\newtheorem{remark}[theorem]{Remark}

\newtheorem{example}{Example}[theorem]

\newcommand{\annot}[2]{\underset{\textcolor{magenta}{#2}}{#1}}

\newcommand{\F}{\mathbb{F}}

\newcommand{\I}{\mathbb{I}}
\newcommand{\N}{\mathbb{N}}

\newcommand{\Boo}{\{0,1 \}}

\newcommand{\Grobner}{\mathfrak{G}}

\newcommand{\bigO}{\mathcal{O}}
\newcommand{\mc}[1]{\mathcal{#1}}
\newcommand{\cF}{\mathcal{F}}
\newcommand{\cV}{\mathcal{V}}
\newcommand{\calf}{\cF}
\newcommand{\cP}{\mathcal{P}}
\newcommand{\corr}{\mathcal{LC}}
\newcommand{\vldt}{\mathcal{LDT}}
\newcommand{\vvanish}{\mathcal{ZERO}}
\newcommand{\x}{\mathbf{x}}

\newcommand{\y}{\mathbf{y}}

\DeclareMathOperator{\LM}{LM}
\DeclareMathOperator{\poly}{poly}

\DeclareMathOperator{\Span}{span}
\DeclareMathOperator{\rank}{rank}

\newcommand{\paren}[1]{\left( #1 \right)}

\newcommand{\set}[1]{\left\{ #1 \right\}}

\newcommand{\reject}{\texttt{Reject}}

\newcommand{\gc}{k}

\DeclareMathOperator*{\E}{\mathbb{E}}

\definecolor{thmcolor}{RGB}{235, 235, 235}
\definecolor{citecolor}{RGB}{1, 210, 56}
\definecolor{lemmacolor}{RGB}{130, 169, 252}

\newtcolorbox{algobox}{colback=lightgray!5!white,colframe=lightgray!75!black}
\newtcolorbox{thmbox}{colback=thmcolor!5!white,colframe=black!75!black}
\newtcolorbox{lemmabox}{colback=lemmacolor!5!white,colframe=blue!75!blue}

\usepackage{setspace}
\usepackage{enumerate}
\usepackage[parfill]{parskip}
\usepackage{changepage}
\usepackage{framed}

\allowdisplaybreaks
\hypersetup{
	colorlinks,
	linkcolor={blue},
	citecolor={citecolor},
	urlcolor={blue}
}

\usepackage[
	backend=biber,
	style=alphabetic,
	sorting=nyt,
	backref=true,
	maxcitenames = 8,
	mincitenames = 5,
	maxalphanames = 8,
	minalphanames = 5,
	maxnames = 10,
	minnames = 5
]{biblatex}

\ifnum\draft=1
	\newcommand{\anote}[1]{{\color{brown} [Amik: #1]}}
	\newcommand{\mnote}[1]{{\color{red} [Madhu: #1]}}
	\newcommand{\pnote}[1]{{\color{blue} [Prashanth: #1]}}
	\newcommand{\snote}[1]{{\color{teal} [Srikanth: #1]}}
	\newcommand{\swnote}[1]{{\color{purple} [Sophus: #1]}}
\else
	\newcommand{\anote}[1]{}
	\newcommand{\mnote}[1]{}
	\newcommand{\pnote}[1]{}
	\newcommand{\snote}[1]{}
	\newcommand{\swnote}[1]{}
\fi

\date{March 24, 2026}

\begin{document}
\title{Ideals, Macaulay Bases, and PCPs}

\if\anon0{\author{Prashanth Amireddy\thanks{School of Engineering and Applied Sciences, Harvard University, Cambridge, Massachusetts, USA. Supported in part by a Simons Investigator Award and NSF Award CCF 2152413 to Madhu Sudan and a Simons Investigator Award to Salil Vadhan. Email: \texttt{pamireddy@g.harvard.edu}} \and
		Amik Raj Behera\thanks{Department of Computer Science, University of Copenhagen, Denmark. Supported by Srikanth Srinivasan's start-up grant from the University of Copenhagen. Email: \texttt{ambe@di.ku.dk} } \and
		Srikanth Srinivasan \thanks{Department of Computer Science, University of Copenhagen, Denmark. Supported by the European Research Council (ERC) under grant agreement no. 101125652 (ALBA). Email: \texttt{srsr@di.ku.dk} } \and
		Madhu Sudan\thanks{School of Engineering and Applied Sciences, Harvard University, Cambridge, Massachusetts, USA. Supported in part by a Simons Investigator Award, NSF Award CCF 2152413 and AFOSR award FA9550-25-1-0112. Email: \texttt{madhu@cs.harvard.edu}} \and Sophus Valentin Willumsgaard \thanks{Department of Computer Science, University of Copenhagen, Denmark. Supported by the European Research Council (ERC) under grant agreement no. 101125652 (ALBA). Email: \texttt{sophus.willumsgaard@di.ku.dk} }   }}\else{
}\fi

\maketitle

\pagenumbering{arabic}

\begin{abstract}
    All known proofs of the PCP theorem rely on multiple ``composition'' steps, where PCPs over large alphabets are turned into PCPs over much smaller alphabets at a (relatively) small price in the soundness error of the PCP. Algebraic proofs, starting with the work of Arora, Lund, Motwani, Sudan, and Szegedy use at least 2 such composition steps, whereas the ``Gap amplification'' proof of Dinur uses $\Theta(\log n)$ such composition steps. In this work, we present the first PCP construction using just one composition step. The key ingredient, missing in previous work and finally supplied in this paper,  is a basic PCP (of Proximity) of size $2^{n^\varepsilon}$, for any $\varepsilon > 0$, that makes $\mathcal{O}_\varepsilon(1)$ queries.

	At the core of our new construction is a new class of alternatives to ``sum-check'' protocols. As used in past PCPs, these provide a method by which to verify that an $m$-variate degree $d$ polynomial $P$ evaluates to zero at every point of some set $S \subseteq \mathbb{F}_q^m$. Previous works had shown how to check this condition for sets of the form $S = H^m$ using $\mathcal{O}(m)$ queries with alphabet $\mathbb{F}_q^d$ assuming $d \geq |H|$.
	Our work improves this basic protocol in two ways: First we extend it to broader classes of sets $S$ (ones closer to Hamming balls rather than cubes). Second, it reduces the number of queries from $\mathcal{O}(m)$ to an absolute constant for the settings of $S$ we consider. Specifically when $S = (\{0,1\}^{m/c}_{\leq 1})^c$, where $T = \{0,1\}^{a}_{\leq b} \subseteq \mathbb{F}_q^a$ denotes the set of Boolean vectors of Hamming weight at most $b$ in $\mathbb{F}_q^a$, we give such an alternate to the sum-check protocol with $\mathcal{O}(1)$ queries with alphabet $\mathbb{F}_q^{\mathcal{O}(c+d)}$, using proofs of size $q^{\mathcal{O}(m^2/c)}$. Our new protocols use the notion of Macaulay bases to extend previously known protocols to these new settings with surprising ease. In doing so, they highlight why these notions from algebra may be of further use in complexity theory.

\end{abstract}

\newpage

\tableofcontents

\newpage

% \pagenumbering{arabic}

\section{Introduction}

In this paper, we give a new general framework for constructing algebraic PCPs that leads to the first proof of the PCP theorem using only one ``composition'' step. Starting with the work of Arora and Safra~\cite{AroraS}, composition of PCPs has been a key ingredient in all previous PCP constructions. The original proof of the PCP theorem due to Arora, Lund, Motwani, Sudan, and Szegedy~\cite{ALMSS} used two composition steps, while the novel alternate proof due to Dinur~\cite{Dinur} uses $O(\log n)$ composition steps. Compositions improve various parameters of the construction at the cost of making the verifier less transparent. So it is a natural and long-sought goal to try to minimize the number of composition steps (to one, or even zero!). We achieve the weaker goal here.

The key to our construction is a new class of protocols replacing the ``sum-check'' protocol in PCP constructions.
The sum-check protocol, due to Lund, Fortnow, Karloff, and Nisan~\cite{LFKN} (also used in proofs of IP=PSPACE ~\cite{Shamir} and MIP=NEXPTIME~\cite{BFL}) has been a central ingredient in previous PCP constructions. In PCP constructions, the protocol is used to establish that an $m$-variate polynomial $P$ over $\F_q$ given as an oracle from $\F_q^m \to \F_q$ is identically zero on the set $\{0,1\}^m \subseteq \F_q^m$ or more generally, on some set of the form $H^m$ for $H \subseteq \F_q$. We refer to this latter task as ``zero-on-variety testing''. (The reason for the use of the term ``variety'' to describe the set $H^m$ will become clearer later.) Other than the sum check protocol, the only other direct protocol for zero-on-variety testing is a protocol due to Ben-Sasson and Sudan~\cite{BenSassonS}, which also only works for varieties of the form $H^m$. In this work, we establish a new connection between the theory of Macaulay bases and the zero-on-variety test of \cite{BenSassonS} that allows us to get efficient zero-on-variety tests for a much broader class of varieties, including some varieties that are close to Hamming balls of constant radii. This latter setting which had eluded previous works and is key to our PCP construction.

Armed with this new class of protocols, we show how to significantly simplify the ALMSS PCP construction. We start by giving a new PCP construction that works relative to any variety $V \subseteq \F_q^m$ with performance depending on the {\em Macaulay complexity} of the variety, which is defined based on the notion of {\em H-bases} (or {\em Macaulay bases}, and is related to the notion of {\em Gr\"obner bases}) which go all the way back to the work of Macaulay in the 1900s, see~e.g.~\cite{sauer2001grobner, macaulay, cox1997ideals, moeller2000h}. We then show how to specialize this PCP in two different ways by using two different varieties --- the first giving the usual $\bigO(\log n)$ randomness and $\poly(\log n)$ query PCP for $\mathsf{NP}$, and the second giving an $\bigO(n^\varepsilon)$ randomness and $\bigO_\varepsilon(1)$ query PCP for NP, for any $\varepsilon > 0$. No natural PCP (built without composition) was known with the latter setting of parameters, and indeed, this has been the key bottleneck in reducing the number of composition steps in ALMSS. We stress that both our ingredient PCPs are instantiations of the same protocol --- only the choice of the variety is different (and the analysis of the Macaulay complexity of the varieties is straightforward)! And furthermore, our PCPs are already ``robust assignment testers'' in the sense of Dinur and Reingold~\cite{DinurReingold} (or equivalently, Robust PCPs of Proximity in the sense of Ben-Sasson, Goldreich, Harsha, Sudan and Vadhan~\cite{BGHSV}) and thus immediately composable. Putting our two robust assignment testers together yields our final PCP. The resulting proof thus gives the following simplifications to the ALMSS protocol: It eliminates one composition step, it eliminates the need for the ``Hadamard PCP'' entirely, and it eliminates the need for the ``parallelization/robustification'' step in ALMSS~\cite[Section 7]{ALMSS}.

In what follows, we describe our work in greater detail, starting with the basic notion of interest in this paper, the zero-on-variety testing problem, and the resulting PCPs.

\subsection{Overview of Our Construction}

\paragraph{PCPs of Proximity} The central objects of interest in this paper are best described by the umbrella term ``PCP of Proximity" (or equivalently ``assignment tester''). Here a verifier $\cV$ is given oracle access to some oracle $f:D \to \Sigma$  along with a proof oracle $\pi: S \to \Gamma$ where $D,S,\Sigma$ and $\Gamma$ are finite sets. A verifier for a property $\calf \subseteq \{g:D \to \Sigma\}$ queries $(f,\pi)$ and renders a verdict \texttt{Accept}/\texttt{Reject}, with the property that if $f \in \calf$ then there exists a $\pi$ such $\cV^{f,\pi}$ always outputs \texttt{Accept}, while for $f \not\in \calf$ we have that for every proof $\pi$, $\cV^{f,\pi}$ outputs \texttt{Reject} with probability\footnote{Here $\delta_\calf(f) = \min_{g\in \calf}\{\delta(f,g)\}$ and $\delta(f,g) = \Pr_{x\in D} [f(x) \ne g(x)]$ measure the distance of $f$ from $\calf$ is the normalized Hamming metric.} $\Omega(\delta_\calf(f))$. The key parameters associated with the verifier are its randomness (usually $\bigO(\log |S|)$), its locality (or query complexity) $\ell$ which is the total number of queries to $f$ and $\pi$, and the alphabet size $a = \max\{\log |\Sigma|,\log |\Gamma|\}$.

A good example of a PCP of Proximity is the low-degree test. Here $D = \F_q^m$ and $\Sigma = \F_q$. The property $\calf_{m,d,q}$ is the set of evaluations of all $m$-variate polynomials over $\F_q$ of total degree at most $d$. When $d\leq q$ the best known low-degree tests achieve randomness of $\bigO(m\log q)$, locality $\ell = 2$, and alphabet size $d \log q$. (See \Cref{thm:low-degree-testing}.) The fact that the locality is a constant and $d$ and $m$ affect only the randomness and alphabet size is important in their use in PCPs.

\paragraph{Zero-on-variety testing} The zero-on-variety testing problem is also a PCP of Proximity problem. It is described by parameters $\F_q,d,m$ and a variety $V \subseteq \F_q^m$. Here the verifier is given oracle access to a function $f:\F_q^m\to\F_q$ that is promised to be a degree $d$ polynomial, and goal is to test for the property $\calf_{V,d}$ that is the set of all degree $d$ polynomials that are identically zero on $V$, or equivalently if the polynomial $f$ lies in the ideal of polynomials $\I(V)$ vanishing on $V$. (In this description, we opt to describe this as a promise problem - though in PCP applications, the non-promised version of this problem is the one used. The two become essentially equivalent thanks to the existence of low-degree tests.)

Prior to this work, the only natural zero-on-variety tests considered the setting where $V = H^m$ for some $H \subseteq \F_q$. The protocol given by \cite{BenSassonS} uses the following identity. $f \in \I(H^m)$ if and only if there exist polynomials $f_1,\ldots, f_m$ of degree at most $d-|H|$ such that
\begin{equation}
	f(X) = \sum_{i=1}^m f_i(X) Z_H(X_i),   \label{eqn:comb-null}
\end{equation}
where $Z_H(Y) = \prod_{\alpha\in H} (Y - \alpha)$ is the canonical univariate polynomial that vanishes on $H$.
The  identity above follows from Alon's Combinatorial Nullstellensatz~\cite{Alon:Null} and leads to a tester as follows:
The zero-on-variety tester for $H^m$ expects oracles for $f_1,\ldots,f_m$ as proof. It verifies using the low-degree test that each of these oracles has degree at most $d-|H|$ and then verifies \Cref{eqn:comb-null} for a random choice of $X = (a_1,\ldots,a_m)$. Modulo further details (involving local correction), this leads to an $\bigO(m)$ local tester with randomness $\bigO(m\log q)$ and alphabet size $\bigO(d \log q)$.

For our purposes, this choice of variety is insufficient. (Furthermore, the dependence of the locality on $m$ is also problematic, but we'll address this later.) To remedy this, we use an alternate interpretation of the identity above. In this interpretation the identity holds because $\{Z_H(X_1),\ldots,Z_H(X_m)\}$ form a ``Gr\"obner basis'' of the ideal $\I(H^m)$ under a ``graded monomial ordering''. We won't define the exact notion of a Gr\"obner basis under different monomial orderings here --- we don't need to. The notion that suffices for us is the notion of a Macaulay basis $\Grobner$ of an ideal $\I$: $\Grobner$ is a {\em Macaulay basis} for $\I$ if for all polynomials $P \in \I$ there exist polynomials $h_g, g \in \Grobner$ such that $P = \sum_{g \in \Grobner} h_g \cdot g$ and $\deg(h_g \cdot g) \leq \deg(P)$ for every $g \in \Grobner$. And the above identity is the special case where $V = H^m$ with $\{Z_H(X_1),\ldots,Z_H(X_m)\}$ as the Macaulay basis. The zero-on-variety test of \cite{BenSassonS} can now be extended to any variety that has a ``small'' Macaulay basis. Indeed, this motivates our notion of the Macaulay complexity of a variety $V$ (see \Cref{defn:macaulay-complexity} for a formal definition) --- which is the size of the smallest Macaulay basis $\Grobner$ of $\I(V)$. One crucial example for us is the following: The set of polynomials $\{X_iX_j | 1 \leq i < j \leq m\} \cup \{X_i^2 - X_i | i\in [m]\}$ form a Macaulay basis for the variety $\{0,1\}^m_{\leq 1}$ consisting of Boolean points of Hamming weight at most $1$ in $\F_q^m$. (Thus the variety $\{0,1\}^m_{\leq 1}$ has Macaulay complexity $\bigO(m^2)$.) We also use the following basic property of Macaulay bases: If $\Grobner_X \subseteq \F_q[X]$ is a Macaulay basis for variety $V_1$ and $\Grobner_Y \subseteq \F_q[Y]$ is one for $V_2$ then $\Grobner_X \cup \Grobner_Y \subseteq \F_q[X,Y]$ is a Macaulay basis for $V_1 \times V_2$. In particular this establishes that the Macaulay complexity of $(\{0,1\}^m_{\leq 1})^c \subseteq \F_q^{mc}$ is $\bigO(cm^2)$.
Applying the \cite{BenSassonS} protocol to Macaulay bases now gives us a $\bigO(k)$ query protocol for testing zero-on-$V$ for variety $V$ of Macaulay complexity $k$.

While this now gives many new varieties that have natural zero-on-variety tests, the locality of $\bigO(k)$ can be prohibitive. Our second contribution here is to give a new protocol to test this, that shifts the Macaulay complexity into the randomness of the protocol and achieves (a universal) constant locality. Specifically our verifier now expects an oracle for $\mathcal{M}_V(P)(X,Y) = \sum_{g \in \Grobner} Y_g h_g(X)$ where $Y = (Y_g | g\in \Grobner)$ is a new set of $k$ variables. Specifically performing a low-degree test on $\mathcal{M}_V(P)$ along with a test that verifies $\mathcal{M}_V(P)(X,0) \equiv 0$ ensures that $\mathcal{M}_V(P)$ is effectively giving access to all linear combinations of $h_g(X)$ with {\em constant} query complexity. Testing the identity $P(X) = \sum_{g\in \Grobner} h_g(X) g(X)$ at a random choice of $X$ now requires only one query to $P$ and one to $\mathcal{M}_V(P)$! Modulo some standard use of self-correction, this gives us an $\bigO(1)$ locality protocol for zero-on-$V$ testing with alphabet size $\bigO(d \log q)$ and randomness $\bigO((k+m)\log q)$ for any variety $V$ of Macaulay complexity $k$. (See \Cref{sec:vanishing-certificate}.)

\paragraph{PCPs from Zero-on-variety tests}

It is straightforward to build PCPs for NP-hard problems from low-degree tests and zero-on-variety tests. (Recall that a PCP verifier for say graph coloring is given as input a graph $G$ and oracle access to a purported proof $\pi$ with the feature that if $G$ is 3-colorable that there exists a $\pi$ such that $V$ always accepts whereas if $G$ is not 3-colorable then $V$ rejects every proof $\pi$ w.p. at least $1/2$. The parameters of interest to us are the same --- the randomness, the locality, and the alphabet size of the proof.)

For example, the 3-coloring protocol, based on a similar proof from \cite{BenSassonS}, goes as follows: Fix an odd prime power $q$.\footnote{We do this for simplicity here and allow us to assume $\{-1,0,1\}\subseteq \F_q$ can be used to represent $3$ distinct colors. The protocol easily extends to other fields using some other set of $3$ distinct elements of $\F_q$.} For a variety $V \subseteq \F_q^m$ let  its ``extension degree'' be the least integer $d$ such that every function $f:V \to \F_q$ can be extended to a degree $d$ polynomial in $\F_q[X_1,\ldots,X_m]$. Now, given a variety $V \subseteq \F_q^m$ of Macaulay complexity $k$ and extension degree $d$ we consider the $3$-coloring problem on the vertex set $V$ (same $V$). Note that the graph is given by an edge function $E:V \times V \to \{0,1\}$ which can be shown to be extendable to a degree $2d$ polynomial $\hat{E}$ from $\F_q^{2m}\to \F_q$. A proof that $E$ is $3$-colorable includes polynomials $\chi:\F_q^m\to\F_q$, $A:\F_q^m\to\F_q$ and $B:\F_q^{2m}\to\F_q$ satisfying (1) $A(X) = \chi(X)\cdot(\chi(X)-1)\cdot(\chi(X)+1)$, (2) $A$ is zero-on-$V$, (3) $B(X,Y) = \hat{E}(X,Y) \prod_{i \in \{-2,-1,1,2\}} (\chi(X)-\chi(Y)-i)$ and (4) $B$ is zero on $V \times V$. (Items (1) and (2) verify that $\chi$ is a $3$-coloring of $V$ with color set $\{-1,0,1\}$,  while items (3) and (4) verify that $\chi$ is a valid coloring of the edges of $E$.) The $V$-verifier performs low-degree tests on all the four oracles and then tests identities (1) and (3) by picking a random value of the variables, and finally verifies items (2) and (4) using zero-on-$V$ and zero-on-$V^2$ tests. By the aforementioned properties of Macaulay complexity and standard facts about extension degree we get that this PCP verifier achieves $\bigO(1)$ locality with randomness $\bigO((k+m)\log q)$ and alphabet $\bigO(d \log q)$ (matching those of the zero-on-$V$ tests up to constant factors).

Instantiating the verifier above with $V = H^m$ where $|V| =n$, $|H|=\log n$, $m = \bigO(\log n/\log \log n)$ and $q = \bigO(\log^6 n)$ gives an $\bigO(1)$ locality PCP verifier for 3-coloring of $n$ vertex graphs with randomness $\bigO(\log n)$ and alphabet size $\bigO(\poly\log n)$. But a different instantiation {\em of the same PCP} with $m = n^\varepsilon$, $c = \frac1\varepsilon$, $q = \bigO_\varepsilon(1)$ and $V=(\{0,1\}^m_{\leq 1})^c$ gives an $\bigO(1)$ locality PCP verifier with randomness $\bigO(n^{2\epsilon})$ and alphabet size $O_\epsilon(1)$! We note that even using $c=1$ gives a completely new protocol matching the parameters of the Hadamard PCP in \cite[Section 5]{ALMSS}. And using larger values of $c$ gives us our new protocols. (See \Cref{sec:main-thm} for details.)

Furthermore, these PCPs are easily converted to ``Robust PCPs of Proximity'' (or ``Robust assignment testers'') in the sense of \cite{DinurReingold,BGHSV} of constant robustness --- since our PCPs have constant locality. This allows us to compose the PCPs above in a single composition step to get an $\bigO(1)$ locality PCP verifier with $\bigO(\log n)$ randomness and $\bigO(1)$ alphabet size --- and thus the PCP theorem. (See \Cref{sec:composition} for details.)

\paragraph{Related constructions.} Here we will discuss two PCP constructions that are relevant to us; both of them attain the parameters of our inner verifiers (i.e., $\bigO(1)$ locality and $\bigO(n^\varepsilon)$ randomness, and $\bigO_\varepsilon(1)$ alphabet size) using constructions that don't make use of a ``full composition'' but do perform certain transformations/encodings of an initial (sub-optimal) proof to achieve the final parameters.

First, we mention the work of Goldreich~\cite{goldreich2025proving} posted shortly after our paper, which gives a different construction of a PCP matching our inner verifier parameters. This construction is based on using a Reed-Muller PCP at an outer level and encoding the larger field elements using a Hadamard code. The tests of the verifier then resemble that of the usual Hadamard PCP of~\cite{ALMSS}. Composing this inner verifier (rather the PCP of Proximity version of it) with the the (robust version of the) outer verifier of~\cite{ALMSS} also results in the PCP Theorem (in 1.5 proof compositions as Goldreich refers to it).

Second, we mention the work of~\cite{khot2006ruling} studying {\em quasi-random PCPs}, which have a very different behavior from the usual PCPs. While~\cite{khot2006ruling} obtains a constant query quasi-random PCP of size similar to our inner verifier (i.e., $2^{n^\varepsilon}$), the construction (among many other differences) involves using the Hadamard encodings of the elements of a large field. We also note that both the works of~\cite{khot2006ruling} as well as~\cite{goldreich2025proving} also deviate from our techniques in that we do not perform a sum-check protocol in order to do the zero-on-variety tests.

\section{Formal Statement of Our Results}

We first introduce basic definitions needed to state our main result. Throughout this document, $\mathcal{V}^{\Pi}$ means that the algorithm/circuit/verifier has oracle access to the string $\Pi$, i.e., $\mathcal{V}$ can query $\Pi[i]$ for any $1 \leq i \leq |\Pi|$. We use the notation $\mathcal{V}^{\Pi}(x;R)$ to say that the algorithm $\mathcal{V}$ has oracle access to $\Pi$, has input $x$, and access to a random string $R$. In this notation, $\mathcal{V}^{\Pi}(x;R)$ is a deterministic algorithm and the randomness is in the choice of $R$.

\subsection{PCPs}

\begin{definition}[Standard Verifier]\label{defn:standard-verifier}
	For functions $r,\ell,a: \mathbb{Z}^{\geq 0} \to \mathbb{Z}^{\geq 0}$, define a $(r,\ell,a)$-standard verifier $\mathcal{V}$ as follows:\newline
	Let $\Sigma = \Boo^{a(n)}$. On input $x \in \Boo^{n}$ of length $n$, a string\footnote{This string $R$ is the random string fed into $\mathcal{V}$.} $R \in \Boo^{r(n)}$, and oracle access to a string $\Pi \in \Sigma^{\mathrm{size}(n)}$ (i.e. $\Pi$ is a string of length $\mathrm{size}(n)$ on alphabet $\Boo^{a(n)}$), we have:\\
	\begin{itemize}
		\item $\mathcal{V}^{\Pi}(x;R)$ outputs a subset $Q \subseteq [\mathrm{size}(n)]$ of cardinality $\ell(n)$.
		\item $\mathcal{V}^{\Pi}(x;R)$ outputs a Boolean circuit $\mathcal{C}$ (depends on $x$ and $R$) depending on $\ell(n) \cdot a(n)$ bits. The circuit $C$ gets access to entries of $\Pi$ as bits of length $a(n)$.
		\item $\mathcal{V}^{\Pi}(x;R)$ returns \texttt{Accept} if $C(\Pi|_{Q}) = 1$ and returns \texttt{Reject} if $C(\Pi|_{Q}) = 0$.
	\end{itemize}

	The maximum circuit size $|\mathcal{C}|$ over every possible choice of $(x,R)$ will be referred to as the size of the standard verifier $\mathcal{V}$. The running time of the standard verifier $\mathcal{V}$ will be $\poly(n\cdot 2^{r(n)})$.
\end{definition}

\noindent
Observe that in the above definition, $\mathcal{V}$ makes $\ell(n)$ queries to $\Pi$ using the $r(n)$ coin tosses. In particular, $\mathcal{V}$ can only query a coordinate within range of $[0, \ell(n) \cdot 2^{r(n)}-1]$. So from now on, we will always assume that the proof size $|\Pi|$ is $\bigO(\ell(n) \cdot 2^{r(n)})$.\\

\noindent
\begin{definition}[The class PCP]\label{defn:the-class-pcp}
	For functions $r,\ell,a \in \mathbb{Z}^{\geq 0} \to \mathbb{Z}^{\geq 0}$, for $c,s \in (0,1)$, define $\mathsf{PCP}_{c,s}[r,\ell,a]$ to be the class of languages $L$ that have a standard $(r,\ell,a)$ verifier with completeness $\geq c$ and soundness $\leq s$, i.e.
	\begin{itemize}
		\item \textbf{Completeness}: For every $x \in L$, there exists a proof $\Pi$ such that $\Pr_R[\mathcal{V}^\Pi(x;R) = \texttt{Accept}] \geq c$.
		\item \textbf{Soundness}: For every $x \notin L$, for every $\Pi$, we have $\Pr_R[\mathcal{V}^\Pi(x;R) = \texttt{Accept}] \leq s$.
	\end{itemize}

\end{definition}

In this paper, we will usually focus on the language $3$-$\mathsf{COLOR}$ of $3$-colorable graphs.\\

\noindent
Our main theorem (proved in \Cref{sec:main-thm}) shows the following:\\

\begin{thmbox}
	\begin{restatable}[Main Theorem]{theorem}{mainthm}\label{thm:main-thm}
		There exist constants $c$, $\ell$ such that the following holds for every $q,m,d,k$ such that $q \geq c d^3$:\newline

		Let $\F_{q}$ be a field of characteristic $\neq 2$ and let $V \subseteq \F_q^m$ have extension degree $d$ and Macaulay complexity $\gc$.
		Then $3$-$\mathsf{COLOR}$ on vertex set $V$ is in
		$\mathsf{PCP}_{1,1/2}[c (\gc+m) \log q, \; \ell, \; c d \log q]$
		with proofs of size $q^{c(\gc+m)}$.\\
	\end{restatable}
\end{thmbox}

\noindent
\begin{remark}
	In \Cref{thm:main-thm}, the assumption on characteristic $\neq 2$ is mostly for clarity in the proofs. We assign the vertices colors from the set $\set{-1,0,1}$, and these are three distinct colors only if the field is of characteristic $\neq 2$. For fields of characteristic $2$ and with at least $3$ elements, one could use colors $\set{a,b,c}$, where $a,b,$ and $c$ are three distinct elements from the field. The proof is essentially the same.\\
\end{remark}
With this setup, we first show the following two results.\\

\begin{lemma}
	\label{lem:polylogn-pcp}
	For every $n$, 3-$\mathsf{COLOR}$ on $n$-vertex graphs is in \break $\mathsf{PCP}_{1,1/2}\left[\bigO(\log n), \; \bigO(1), \; \bigO(\log^2n /(\log \log n))\right]$ with proofs of size $n^{\bigO(1)}$.
\end{lemma}
\begin{proof}
	In \Cref{cor:Hhamm-wt-c} we show that if $V = H^m$ for some $H \subseteq \F_q$ then it has extension degree $(|H|-1)\cdot m$ and Macaulay complexity $m$. Taking $q \geq c \log^6 n$ a power of $3$, $V=H^m$ for some subset $H\subseteq \F_q$ of size $\log n$ and $m=\log n/\log \log n$, we get the desired bounds.
\end{proof}

\noindent
\begin{lemma}
	\label{lem:n2-pcp}
	For every $n$, 3-$\mathsf{COLOR}$ on $n$-vertex graphs is in \break $\mathsf{PCP}_{1,1/2}\left[\bigO(n^2), \; \bigO(1), \; \bigO(1)\right]$ with proofs of size $2^{\bigO(n^2)}$.
\end{lemma}
\begin{proof}
	In \Cref{example:hammingweightone} we show that if $V = \{0,1\}^m_{\leq 1} := \left\{(a_1,\ldots,a_m) \in \{0,1\}^m | \sum_{i=1}^m a_i \leq 1\right\}$ is the set of Boolean points in $\F_q^m$ of Hamming weight at most $1$, then $V$ has extension degree $1$ and Macaulay complexity $\bigO(m^2)$. Picking $q$ to be a large constant, and $V=\{0,1\}^m_{\leq 1}$ we get the desired bounds.
\end{proof}
Note that the above roughly matches the parameters of the Hadamard PCP of \cite{ALMSS} with a completely different proof!\\

\begin{lemma}
	\label{lem:neps-pcp}

	For every $\varepsilon > 0$ and every $n$, 3-$\mathsf{COLOR}$ on $n$-vertex graphs is in \break $\mathsf{PCP}_{1,1/2}\left[\bigO(n^\varepsilon), \; \bigO(1), \; \bigO\left(\frac{1}{\varepsilon}\log \frac{1}{\varepsilon}\right)\right]$ with proofs of size $2^{\bigO(n^\varepsilon)}$.

\end{lemma}
\begin{proof}
	In \Cref{cor:hamm-wt-c} we show that if $V = \left(\{0,1\}^m_{\leq 1}\right)^c$, then $V$ has extension degree $c$ and Macaulay complexity $\bigO(cm^2)$. Given $\varepsilon > 0$ picking $c=\bigO(\frac1\varepsilon)$, $q = \poly(1/\varepsilon)$, $m=n^{\bigO(\varepsilon)}$, and $V=\left(\{0,1\}^m_{\leq 1}\right)^c$ we get the desired bounds.
\end{proof}

The above concludes the description of the atomic PCPs we construct. In \Cref{sec:composition} we show that these PCPs can be strengthened to ``Robust assignment testers'' (see \Cref{def:testers}), and so can be composed together (see \Cref{lem:DRcomposition}) to get the PCP theorem stated below (proved in \Cref{sec:composition}).\\

\begin{theorem}[PCP Theorem]\label{thm:pcp-thm}
	There exist universal constants $\ell$, $a$, $C$ such that for every $n$, $3$-$\mathsf{COLOR}$ on $n$-vertex graphs is in
	$\mathsf{PCP}_{1,1/2}[C \log n, \ell, a]$.
\end{theorem}

%\onehalfspacing

\section{Preliminaries}

\paragraph*{}For a field $\F_{q}$, we will use $\F_{q}[x_{1},\ldots,x_{m}]$ to denote the multivariate polynomial ring in variables $x_{1},\ldots,x_{m}$. For a degree parameter $d \in \mathbb{N}$, we will use $\mathcal{P}_{d}(\F_{q}^{m}) \subset \F_{q}[x_{1},\ldots,x_{m}]$ to denote the subspace of degree $\leq d$ polynomials. For a polynomial $P \in \F_{q}[x_{1},\ldots,x_{m}]$ and a set $V \subseteq \F_{q}^{m}$, we denote the restriction of $P$ to $V$ by $P|_{V}$. We will denote by $\F_{q}^{\times}$ the set of invertible elements of $\F_{q}$, i.e. $\F_{q}^{\times} = \F_{q} \setminus \set{0}$.\\

\begin{fact}[Polynomial Distance Lemma]\label{thm:odlsz}
	\cite{ore1922hohere, DL78, Schwartz80, Zippel79}.
	Fix a field $\F_{q}$. For every degree parameter $d \in \mathbb{N}$ with $d \leq q$ and for every non-zero polynomial $P \in \F_{q}[x_{1},\ldots,x_{m}]$, we have:
	\begin{align*}
		\Pr_{\mathbf{a} \sim \F_{q}^{m}}[P(\mathbf{a}) = 0] \; \leq \; \dfrac{d}{q}.
	\end{align*}
\end{fact}

\noindent
An immediate and useful corollary of \Cref{thm:odlsz} is the following: If two degree \(\leq d\) polynomials $P$ and $Q$ agree on strictly more than $d/q$-fraction of $\F_{q}^{m}$, then $P = Q$.

\paragraph*{}For any $c,m \in \mathbb{N}$ with $c \leq m$, we use $\Boo^{m}_{\leq c}$ to denote the set of Boolean strings of Hamming weight $\leq c$. We say that two functions $f,g:S\rightarrow T$ are $\delta$-close or $\delta$-far if they differ on at most or at least a $\delta$-fraction of their inputs, respectively.

\paragraph*{Lines Table}For every $m \in \mathbb{N}$, field $\F_{q}$, and points ${\bf a},{\bf b}\in\F_q^m$,
let $\ell_{{\bf a},{\bf b}}:\F_q \to \F_q^m$ (read as ``line passing through ${\bf a}$ with slope ${\bf b}$'') be defined as $\ell_{{\bf a},{\bf b}}(t):={\bf a}+t {\bf b}$.\\

\begin{definition}[Lines Table]\label{defn:lines-table}
	Fix a field $\F_{q}$.
	Let $d \in \N$ be the degree parameter and $m \in \N$ be the number of variables.
	For every degree \(\leq d\) polynomial $f: \F_{q}^{m} \to \F_{q}$,
	we define the $d^{th}$ lines table for $f$
	$f_{\mathrm{lines}}^{(d)} \; : \; \F_q^{2m} \;
		\longrightarrow \; (\F_{q})^{d+1}$ as the function that maps an input $({\bf a},{\bf b})\in \F_q^{2m}$ to  $f(\ell_{{\bf a},{\bf b}}(t))$, where $t$ is an indeterminate. We note that $f(\ell_{{\bf a},{\bf b}}(t))$ is indeed a univariate degree $d$ polynomial in $t$ and can be specified by the $d+1$ coefficients of $t^0,t^1,\ldots, t^d$.
\end{definition}

\paragraph*{3-colorability}We state the $3$-colorability language below. Note that the choice of $3$-coloring as an $\mathsf{NP}$-complete problem instead of one of many others is simply a matter of convenience.\\

\begin{definition}\label{defn:3-color}
	The decision problem $3$-$\mathsf{COLOR}$ is the following problem:\newline
	Given a graph $G = (V,E)$ with $n$ vertices, decide whether there exists a proper coloring of $G$ using $3$ colors, i.e. for every edge $(u,v) \in E$, the vertices $u$ and $v$ are assigned different colors.
\end{definition}

\noindent
\begin{lemma}
	The decision problem $3$-$\mathsf{COLOR}$ is $\mathsf{NP}$-complete.
\end{lemma}

\subsection{Low-degree Testing}
In this subsection, we discuss the standard point-vs-line test for low-degree testing from \cite{ALMSS}. We start by recalling the test and state its properties in \Cref{thm:low-degree-testing}. In the following discussion, we will switch between a function $f: \F_{q}^{m} \to \F_{q}$ and its evaluation vector $f \in (\F_{q})^{q^{m}}$, as both are equivalent.\\

\begin{algobox}
	\begin{algorithm}[H]
		\caption{Low-Degree Test $\vldt^{(\cdot)}_{\cdot}$}
		\label{algo:low-degree-test}

		\DontPrintSemicolon
		\KwIn{Degree parameter $d$, randomness $\mathbf{a},\mathbf{b} \in \F_{q}^{m}$, $t \in \F_{q}^{\times}$, and oracle access to $(f,f')$ where $f \in (\F_{q})^{q^{m}}$ and $f' \in (\F_{q}^{d+1})^{q^{2m}}$ }

		Query $f'[({\bf a},{\bf b})]$ and query $f[{{\bf a}+t{\bf b}}]$
		\tcp*{Two queries to $(f,f')$}

		\vspace{2mm}

		\If{$f'[({\bf a},{\bf b})](t) \neq f[{{\bf a}+t{\bf b}}]$ \tcp*{Running time is $\poly(m,d)$}}{
		\vspace{3mm}
		\Return{\reject}
		\Else{
			\Return{\texttt{Accept}}
		}
		}

	\end{algorithm}
\end{algobox}

\noindent
\begin{theorem}[Low-degree Testing (see for instance {\protect \cite[Theorem A5]{ALMSS}})]\label{thm:low-degree-testing}
	There exists absolute constants \(0 < C,\delta_{0}\)
	such that for every	\(\delta < \delta_{0}\),
	for every $d,q \in \mathbb{N}$ with $q > C d^{3}$, the following holds over $\F_{q}$:
	\begin{enumerate}
		\item If $f \in \mathcal{P}_{d}(\F_{q}^{m})$, then $\vldt^{f,f_{\mathrm{lines}}^{(d)}}(;\mathbf{a},\mathbf{b},t)$ returns \texttt{Accept} with probability $1$ over the random choice of $(\mathbf{a},\mathbf{b},t)$.
		\item For every $f: \F_{q}^{m} \to \F_{q}$ and for every $f' : \F_{q}^{2m} \to \F_{q}$, we have:
		      \begin{align*}
			      \Pr_{\mathbf{a},\mathbf{b},t}[\vldt^{f,f'}_{d}( ; \mathbf{a},\mathbf{b},t) \; \text{returns \texttt{Reject}}] \; \leq \; \delta \;
			      \implies \; \delta(f, \, \mathcal{P}_{d}(\F_{q}^{m})) \; \leq \; 4 \delta.
		      \end{align*}
	\end{enumerate}
	Furthermore, $\vldt^{f,f'}_{d}$ makes $2$ oracle queries, uses
	$\bigO(m \log q)$ bits of randomness, and runs in time $\poly(m,d)$.\\
\end{theorem}

% \noindent
% Note that \Cref{thm:low-degree-testing} also says that for a degree $d$ polynomial $P$ and its lines table $P_{\mathrm{lines}}^{(d)}$, the verifier $\vldt^{P,P_{\mathrm{lines}}^{(d)}}_{d}$ returns \texttt{Accept} with probability $1$.

\begin{remark}
	For low-degree testing, there has been a long line of work on achieving better parameters in terms of field size and soundness guarantee. We refer the interested reader to \cite[Section 1]{HKSS-LDT} for a detailed overview of the results of low-degree testing and also for the state-of-the-art parameters (see \cite[Theorem 1.2]{HKSS-LDT}). We use the low-degree testing from \cite{ALMSS} because the algorithm and analysis are done using the lines table.
\end{remark}

\subsection{Local Correction of Low-Degree Polynomials}
In this subsection, we discuss the local correction/self-correction algorithm for degree $d$ polynomials over $\F_{q}^{m}$ from \cite{ALMSS}. We first describe the local corrector and then analyze it in \Cref{thm:local-correction}. In the following discussion, we will switch between a function $f: \F_{q}^{m} \to \F_{q}$ and its evaluation vector $f \in (\F_{q})^{q^{m}}$, as both are equivalent.

\begin{algobox}
	\begin{algorithm}[H]
		\caption{Local Corrector $\corr^{(\cdot)}_{\cdot}$}
		\label{algo:local-correction}

		\DontPrintSemicolon
		\KwIn{Degree parameter $d$, evaluation point $\mathbf{a} \in \F_{q}^{m}$, randomness $\mathbf{b} \in \F_{q}^{m}$, $t \in \F_{q}^{\times}$, and oracle access to $(f,f')$  where $f \in (\F_{q})^{q^{m}}$ and $f' \in (\F_{q}^{d+1})^{q^{2m}}$ }

		\vspace{2mm}

		Query $f'[(\mathbf{a},\mathbf{b})](t)$ and query $f[{{\bf a}+t{\bf b}}]$ \tcp*{Two queries to $(f,f')$}

		\vspace{2mm}

		\If{$f'[({\mathbf{a},\mathbf{b}})](t) \neq f[\ell_{\mathbf{a},\mathbf{b}}(t)]$ \tcp*{Running time is $\poly(m,d)$}}{

		\vspace{2mm}

		\Return{\texttt{Reject}}
		}
		\Return{$f'[(\mathbf{a},\mathbf{b})](0)$}
	\end{algorithm}

\end{algobox}

\noindent
\begin{theorem}[Local Correction (see e.g.~\cite{ALMSS}, Proposition 7.2.2.1)]\label{thm:local-correction}
	There exists an absolute constant $C > 0$ such that for every $d,q \in \mathbb{N}$ satisfying $q > Cd$, the following holds.\newline

	\begin{enumerate}
		\item If $f$ is a polynomial of degree $d$, then for every $\mathbf{a} \in \F_{q}^{m}$, $\corr^{f,f_{\mathrm{lines}}^{(d)}}(\mathbf{a};\mathbf{b},t)$ is equal to $f(\mathbf{a})$ with probability $1$ over the random choice of $(\mathbf{b},t)$.

		\item Let $f: \F_{q}^{m} \to \F_{q}$ be any function with the condition that there exists a degree $d$ polynomial $P$ such that $\delta(f,P) \leq \delta$. Then for every $f' : \F_{q}^{2m} \to \F_{q}^{d+1}$, for every $\mathbf{a} \in \F_{q}^{m}$, we have:\\
		      If $\corr^{(f,f')}_{d}(\mathbf{a})$ does not return \texttt{Reject}, then $\corr^{(f,f')}(\mathbf{a})$ computes $P(\mathbf{a})$ exactly with high probability over the random choice of $(\mathbf{b},t)$, i.e.
		      \begin{align*}
			      \Pr_{\mathbf{b},t}[\corr^{(f,f')}_{d}(\mathbf{a};\mathbf{b},t) \; = \; P(\mathbf{a}) \quad \text{OR} \quad \corr^{(f,f')}_{d}(\mathbf{a};\mathbf{b},t) \; \text{returns \texttt{Reject}} ]
			      \; \geq \;
			      1 - 2\sqrt{\delta} - \dfrac{d}{q-1}.
		      \end{align*}
	\end{enumerate}

	Furthermore, $\corr^{(f,f')}_{d}(\mathbf{a})$ makes $2$ oracle queries, uses $\bigO(m \log q)$ bits of randomness, and runs in time $\poly(m,d)$.
\end{theorem}

\section{Macaulay Bases}\label{sec:generatingset} In \Cref{sec:vanishing-certificate},
we define a test to check if a polynomial vanishes on a subset \(V \subseteq \F_q^m\).
In this section, we introduce the relevant parameters of such subsets,
which we use to describe the efficiency of such tests.
We also show that these are well-behaved under Cartesian products. We first define the parameters.\\

\begin{definition}[Extension degree]\label{defn:low-deg-extension}
	For a non-empty set $V \subseteq \F_{q}^{m}$,
	function $f:V \to \F_q$ and polynomial $P \in \F_q[X_1,\ldots,X_m]$
	we say $P$ \emph{extends} $f$ if for every $a \in V$, we have $f(a) = P(a)$.\newline
	We define the \emph{extension degree} of $V$ to be the smallest integer $d \in \mathbb{N}$ such that every function $f:V \to \F_q$ can be extended to a polynomial $\widehat{f}$ of total degree at most $d$.
\end{definition}

\noindent
\begin{definition}[Macaulay Complexity]\label{defn:macaulay-complexity}
	For an ideal
	$\I \subseteq \F_q[x_1,\ldots,x_m]$
	we say that a finite set
	$\Grobner \subseteq \I$
	is a Macaulay basis	of $\I$, if every polynomial $P \in \I$
	can be written as follows:
	\begin{align*}
		P = \sum_{g \in \Grobner} h_g \cdot g, \quad \text{ where } \, h_g \in \F_q[x_1,\ldots,x_m] \text{ and } \deg(h_{g} g) \leq \deg(P).
	\end{align*}

	For a set $V \subseteq \F_{q}^{m}$, let $\mathbb{I}(V)$ denote the ideal of polynomials that vanish on $V$.
	We define the \emph{Macaulay complexity} of $V$ to be the cardinality of the smallest Macaulay basis of $\mathbb{I}(V)$.
\end{definition}

A few remarks about Macaulay bases follow below.

\noindent
\begin{remark}
	Note that a Macaulay basis need not be a ``minimal'' one, i.e., it may be possible to add polynomials to a Macaulay basis while still keeping it a Macaulay basis.
	We also abuse notation slightly and refer to a Macaulay basis of a {\em variety} $V$ as a Macaulay basis of the corresponding ideal $\mathbb{I}(V)$.
	A Gröbner basis in a graded ordering is always a Macaulay basis.
	However a minimal Gr\"obner basis is not necessarily a minimal Macaulay basis, see \Cref{rem:mingrob}.
\end{remark} 

\noindent
\begin{example}\label{example:onedim}
	Let \(H\) be a subset of \(\F_q\).
	Then any polynomial in \(\F_q[x]\) which vanishes on \(H\) is divisible by
	\begin{align*}
		\prod_{h \in H}(x-h).
	\end{align*}
	Furthermore, for any function \(f: H \to \F_q\)
	we can find a degree \(|H|-1\) polynomial extending \(f\).
	It follows \(H\) has Macaulay complexity 1,
	and extension degree \(|H|-1\).
\end{example}

%\noindent
%\begin{example}\label{example:hammingweightone}
%	Let \(\{0,1\}^n_{\leq 1} \subseteq \F_q^n\)
%	be the subset of \(\{0,1\}^n\)
%	consisting of points with Hamming weight at most one.
%	Then
%	\begin{align*}
%		\left\{ x_{i}^{2}-x_{i}\right\}_{1 \leq i \leq n}
%		\cup
%		\left\{
%		x_{1}x_{2}, \dots, x_{n-1}x_{n}\right\}_{1 \leq i<j \leq n}
%	\end{align*}
%	is a Macaulay basis.
%	To see this, note that any polynomial \(P\)
%	can be written in the form
%	\begin{align*}
%		P(\x) = \sum_{i}^{n} \; h_{i}(\x) \cdot (x_{i}^{2}-x_{i}) \; + \;
%		\sum_{i<j}^{n} \; g_{i,j}(\x) \cdot (x_{i}x_{j}) \; + \;
%		\ell(\x)
%	\end{align*}
%	where \(\ell\) is a linear function, and
%	\(\deg(h_{i}),\deg(g_{i,j}) \leq \deg(P)-2\).
%	The first two summands vanish on \(\{0,1\}^{n}_{\leq 1}\),
%	so \(P\) vanishes on \(\{0,1\}^{n}_{\leq 1}\)
%	if \(\ell\) also vanishes.
%	However,
%	if \(\ell\) has a non-zero coefficient for any variable $x_i$,
%	then \(\ell\) takes different values on two points which only differ in
%	the \(i\)-th coordinate.
%	It follows that \(\ell\)
%	only vanishes when it is the zero polynomial,
%	and so \(P\)
%	vanishes on \(\{0,1\}^{n}_{\leq 1}\) if and only if
%	it can be written on the form
%	\begin{align*}
%		P(\x) = \sum_{i}^{n} \; h_{i}(\x) \cdot (x_{i}^{2}-x_{i}) \; + \;
%		\sum_{i<j}^{n} \; g_{i,j}(\x) \cdot (x_{i}x_{j}).
%	\end{align*}
%	It follows that \(\{0,1\}^{n}_{\leq 1}\)
%	has Macaulay complexity at most \(\frac{n(n+1)}{2}\)
%	and extension degree \(1\).
%\end{example}

\noindent
\begin{example}\label{example:hammingweightone}
	Let \(\{0,1\}^n_{\leq 1} \subseteq \F_q^n\)
	be the subset of \(\{0,1\}^n\)
	consisting of points with Hamming weight at most one.
	Then
	\begin{align*}
		\left\{ x_{i}^{2}-x_{i}\right\}_{1 \leq i \leq n}
		\cup
		\left\{
		x_{1}x_{2}, \dots, x_{n-1}x_{n}\right\}_{1 \leq i<j \leq n}
	\end{align*}
	is a Macaulay basis.
	To see this, note that any polynomial \(P\)
	can be written in the form
	\begin{align*}
		P(\x) = \sum_{i}^{n} \; h_{i}(\x) \cdot (x_{i}^{2}-x_{i}) \; + \;
		\sum_{i<j}^{n} \; g_{i,j}(\x) \cdot (x_{i}x_{j}) \; + \;
		\ell(\x)
	\end{align*}
	where \(\ell\) is a linear function, and
	\(\deg(h_{i}),\deg(g_{i,j}) \leq \deg(P)-2\).
	The first two summands vanish on \(\{0,1\}^{n}_{\leq 1}\),
	so \(P\) vanishes on \(\{0,1\}^{n}_{\leq 1}\)
	if \(\ell\) also vanishes.
	However,
	if \(\ell\) has a non-zero coefficient for any variable $x_i$,
	then \(\ell\) takes different values on two points which only differ in
	the \(i\)-th coordinate.
	It follows that \(\ell\)
	only vanishes when it is the zero polynomial,
	and so \(P\)
	vanishes on \(\{0,1\}^{n}_{\leq 1}\) if and only if
	it can be written on the form
	\begin{align*}
		P(\x) = \sum_{i}^{n} \; h_{i}(\x) \cdot (x_{i}^{2}-x_{i}) \; + \;
		\sum_{i,j}^{n} \; g_{i,j}(\x) \cdot (x_{i}x_{j}).
	\end{align*}
	It follows that \(\{0,1\}^{n}_{\leq 1}\)
	has Macaulay complexity at most \(\frac{n(n+1)}{2}\)
	and extension degree \(1\).
\end{example}

The above example can be extended to Hamming balls of radius greater than 1.

\noindent
\begin{example}\label{example:ball}
	Let \(\{0,1\}^m_{\leq c} \subseteq \F_q^m\)
	be the subset of \(\{0,1\}^m\)
	consisting of points with Hamming weight at most $c$.
	\begin{align*}
		\left\{x_{i}^{2}-x_{i}\right\}_{1 \leq i \leq n}
		\cup
		\left\{x_{i_{1}}x_{i_{2}} \cdots x_{i_{c+1}}\right\}_{1 \leq i_{1} < \cdots < i_{c+1} \leq n}
	\end{align*}
	is a Macaulay basis.
	To see this, any polynomial can be written on the form
	\begin{align*}
		P(\x) = \sum_{i}^{n} \; h_{i}(\x) \cdot (x_{i}^{2}-x_{i}) \; + \;
		\sum_{i_{1} < \ldots < i_{c+1}}^{n} \; g_{i,j}(\x) \cdot (x_{i_{1}}x_{i_{2}} \cdots x_{i_{c+1}}) \; + \;
		m(\x)
	\end{align*}
	where \(m(\x)\) is a multilinear polynomial of degree at most \(c\),
	and
	\(\deg(h_{i})\leq \deg(P)-2, \deg(g_{i,j}) \leq \deg(P)+1-c\).

	Since the first two summands vanish on \(\{0,1\}^m_{\leq c}\),
	any polynomial is therefore equivalent to a multilinear polynomial of
	degree at most \(c\) on \(P\),
	so multilinear polynomials span the set of functions on
	\(\left\{0,1\right\}_{\leq c}^{m} \).
	However, since both the size of \(\{0,1\}^m_{\leq c}\) and the number of
	multilinear monomials of degree at most \(c\) is \(\)
	\(\sum_{i=0}^{c}\binom{m}{i}\),
	it follows that a multilinear polynomial can only vanish on \(V\) if it is
	0. It follows that \(f_{|V}=0\) if and only if \(m(\x) = 0\) if and only
	if \(f\) lies in the Macaulay basis.

	This set then has Macaulay complexity
	$m+{m\choose c+1} =\Theta_c(m^{c+1})$ and extension degree $c$.
\end{example}

The following lemma shows that we can upper-bound both the Macaulay complexity and the extension degree
of Cartesian products.

\noindent
\begin{restatable}[Subadditivity of Macaulay complexity and extension degree]{lemma}{subaddgrob}\label{lemma:subaddgrob}
	Let \(V_{1} \subseteq  \F_{q}^{m_{1}}\)
	and \(V_{2} \subseteq  \F_{q}^{m_{2}}\),
	and consider their product
	\(V_{1} \times V_{2} \subseteq \F_{q}^{m_{1}+m_{2}}\).
	We then have:
	\begin{enumerate}
		\item if \(\Grobner_1,\Grobner_2\) are Macaulay bases for
		      \(\I(V_{1}),\I(V_{2})\) respectively,
		      then \(\Grobner_1 \cup \Grobner_2\) is a Macaulay basis for \(\I(V_{1} \times V_{2})\).
		\item if \(V_{1},V_{2}\) have extension degrees \(d_{1}, d_{2}\) respectively,
		      then \(V_{1} \times V_{2}\) has extension degree at most \(d_{1}+d_{2}\).
	\end{enumerate}
	In particular,
	both the Macaulay complexity and extension degree
	are subadditive under Cartesian products.
\end{restatable}

\begin{proof}[Proof of \Cref{lemma:subaddgrob}]
	Let \(V\) be a subset of \(\F_{q}^{m}\).
	We first note,
	that we can a find a monomial basis \(S \subseteq \F_{q}[\x]\)
	for the space of functions \(\F_{q}^{V}\),
	such that any polynomial is equivalent to a linear combination of monomials
	from \(S\) of same or lesser degree.

	We argue as follows. Since \(V\) is finite,
	\(\F_{q}^{V}\) is spanned by polynomials and so is also spanned by monomials.
	Then we can create \(S\) iteratively by degree,
	by first setting \(S_{0} = \{1\}\) as a basis of
	\(\F_{q}^{V} \cap \cP_{\leq 0}(\F_{q}^{m})\),
	and \(S_{i+1}\) by extending \(S_{i}\) to a monomial basis of
	\(\F_{q}^{V} \cap \cP_{\leq i+1}(\F_{q}^{m})\),
	and setting \(S = \bigcup_{i} S_{i}\).
	Since any polynomial of degree \(i\) is contained in the span of \(S_{i}\),
	it must then be equivalent to a sum of monomials of degree at most \(i\),
	showing the desired property.

	Now let \(S\) be such a basis.
	Then any polynomial \(P\)
	as a function from \(V\) to \(\F_q\)
	is equivalent to a linear sum
	\begin{align*}
		P \equiv \sum_{s \in S} c_s s
	\end{align*}
	and so
	\begin{align*}
		P - \sum_{s \in S} c_s s \equiv 0
	\end{align*}
	is a polynomial of degree at most \(\deg(P)\) vanishing on \(V\).
	It follows that a subset
	\(\Grobner \subseteq \F_{q}[\x]\)
	is a Macaulay basis of \(\I(V)\)
	if and only if every polynomial \(P\)
	can	be written in the form
	\begin{align*}
		P =\sum_{s \in S} c_{s} s +  \sum_{g \in \Grobner} h_{g}g,
	\end{align*}
	where each non-zero summand has degree at most \(\deg(P)\).

	Now,
	let \(V_{1},V_2\) be subsets
	with \(S_{1}, \Grobner_1 \subseteq \F_{q}[\x]\)
	and
	\( S_{2}, \Grobner_2 \subseteq \F_{q}[\y]\) as above,
	and set
	\begin{align*}
		S_{12}:=
		\{ s_{1}s_{2} \; \vert \;
		s_{1} \in S_{1}, \; s_{2} \in S_{2}\}.
	\end{align*}
	Then \(S_{12}\) is a basis for functions
	\(V_{1} \times V_{2} \to \F_{q}\) with the above mentioned
	property.
	We will show that \(\Grobner_1 \cup \Grobner_2\)
	is a Macaulay basis for $V_1\times V_2$.
	Let \(m_{1}(\x)m_{2}(\y)\) be a monomial in \(\F_{q}[\x,\y]\),
	We can apply the above property to each monomial separately to get
	\begin{align*}
		m_{1}(\x)m_{2}(\y) =
		\left(
		\sum_{s \in S_{1}} c_{s}s +
		\sum_{g \in \Grobner_{1}} h_{g}g
		\right)
		\left(
		\sum_{s' \in S_{2}} c_{s'}s' +
		\sum_{g' \in \Grobner_{2}} h_{g'}g'
		\right) \\ \\
		= \sum_{ss' \in S_{12}} c_{s}c_{s'}ss' \; + \;
		\annot{\sum_{g \in \Grobner_{1}}  h_{g}g}{\leq\deg(m_1)} \;
		\left(
		\annot{\sum_{s' \in S_{2}} c_{s'}s'}{\leq\deg(m_2)} +
		\annot{\sum_{g' \in \Grobner_{2}} h_{g'}g'
		}{\leq\deg(m_2)}\right) \;
		+ \;
		\annot{\sum_{g' \in \Grobner_{2}} h_{g'}g'}{\leq \deg(m_2)} \;
		\annot{\left( \sum_{s \in S_{1}} c_{s}s\right)}{\leq \deg(m_1)}
		.
	\end{align*}

	We see this gives a representation of the monomial as a linear combination of terms from \(S_{12}\) and \(\Grobner_1 \cup \Grobner_2\)
	of degree at most \(\deg(m_1)+\deg(m_2)\).
	Since we can write a polynomial as a sum of monomials, we get that \(\Grobner_1 \cup \Grobner_2\) is a Macaulay basis. This proves the first item of the lemma.

	To show the second item of the lemma, note that any function
	\(f: V_{1} \times V_{2} \to \F_{q}\)
	can be written as a finite sum
	\begin{align*}
		f(\x,\y) = \sum_{i=1}^{n}
		a_{i}(\x) \cdot b_{i}(\y)
	\end{align*}
	where \(a_{i},b_{i}\) are functions from \(V_{1},V_{2}\)
	respectively to \(\F_{q}\).
	Since all the functions \(a_{i}\) can be
	represented as degree \(d_{1}\)-polynomials and likewise for \(b_{i}\),
	the above sum gives a representation of \(f\) as a polynomial of degree at
	most \(d_{1}+d_{2}\).
\end{proof}
\begin{corollary}\label{cor:Hhamm-wt-c}
	For any subset \(H \subseteq \F_{q}\), \(H^{m} \subseteq  \F_{q}^{m}\) has Macaulay
	complexity at most \(m\) and extension degree at most \(\left(|H|-1\right)\cdot m\).
\end{corollary}
\begin{proof}
	Combine \Cref{lemma:subaddgrob}
	and \Cref{example:onedim}.
\end{proof}
\begin{corollary}\label{cor:hamm-wt-c}
	For $c\in N$ that divides $m$, the subset \(\left(\{0,1\}^{m/c}_{\leq 1}\right)^{c} \subseteq \F_{q}^{m}\)
	has Macaulay complexity at most \(\frac{m^{2}+ cm}{2c}\)
	and extension degree at most \(c\).
\end{corollary}
\begin{proof}
	Combine \Cref{lemma:subaddgrob}
	and \Cref{example:hammingweightone}.
\end{proof}
\begin{remark}
	The sets given by~\Cref{cor:hamm-wt-c} and~\Cref{example:ball} have similar extension degree (i.e., $c$) and size (i.e., $\Theta_c(m^c)$). We note however that the Macaulay complexity bound for the latter family grows with $c$ as $\Theta(m^{c+1})$ (unlike the former family which is always $O(m^2)$), and so the set given by~\Cref{example:ball} does not yield PCPs matching the parameters we get by using~\Cref{cor:hamm-wt-c} instead.
\end{remark}

The following two lemmas show that given a subset \(V\) of \(\F_{q}^{m}\),
we can efficiently compute its extension degree,
find a polynomial extending a function \(f: V \to \F_{q}\),
and compute a minimal Macaulay basis.
\begin{lemma}[Computing extension degree and low-degree extensions]
	\label{lem:alg-extension-degree}
	For every prime power $q$ and $m \in \N$, there exists:
	\begin{enumerate}

		\item An algorithm, which takes a set of \(n\) points \(V\) in \(\F_{q}^{m}\)
		      as input and gives the extension degree \(d\) as output in time
		      \(\poly\left(n,\log q, \binom{m+d}{m}\right)\).
		\item An algorithm, which takes a set of \(n\) points \(V\) in \(\F_{q}^{m}\)
		      and a function \(f: V \to \F_{q}\) as input
		      and gives a polynomial which extends \(f\) as output in time
		      \(\poly\left(n,\log q, \binom{m+d+1}{m}\right)\).
	\end{enumerate}
\end{lemma}
\begin{proof}
	We first show \(1.\)
	\paragraph{Computing extension degree}
	For every \(i\), fix an ordering of the \(\binom{m+i}{i}\)
	monomials of degree at most	\(i\) and also an ordering of the points in $V.$

	Then for every \(i\),
	we can calculate the evaluation matrix \(E_{i}\)
	of dimension \(n \times \binom{m+i}{m}\) where
	\begin{align*}
		(E_{i})_{(jk)} =
		m_{k}(v_{j}),
	\end{align*}
	$v_j$ is the $j$th point of $V$  and $m_k$ the $k$th monomial.
	Then if the rank of \(E_{i}\) is equal to \(n\),
	the monomials of degree at most \(i\)
	span the functions on \(V\),
	and so the extension degree is \(i\).

	We can then find the extension degree
	by calculating the rank of \(E_{i}\) for every \(i\), until \(E_{i}\) has rank \(n\).
	Every \(E_{i}\) is a submatrix of \(E_{i+j}\)
	so we can reuse the computation for every \(i\).
	It follows that the algorithm performs Gaussian elimination on a single
	\(n \times \binom{m+d}{d}\) matrix.
	\paragraph{Finding extending polynomial}
	We first find the extension degree using the previous step.
	Then since \(E_{d}\) has rank \(n\),
	we can find a right inverse \(A\) such that \(E_{d}\cdot A = I_{n}\).
	If we represent the function \(f: V \to \F_{q}\)
	as a \(n\)-dimensional vector \(y\),
	then \(A \cdot y\) gives a vector in the monomial basis,
	which represents a polynomial extending \(y\) since
	\begin{align*}
		E_{d} \cdot \left(A \cdot y\right) = \left(E_{d} \cdot A\right) \cdot y = y.
	\end{align*}
	It follows that the algorithm performs
	one calculation of a right inverse,
	and one matrix-vector multiplication.
\end{proof}
\begin{lemma}[Computing a smallest Macaulay basis]
	\label{lem:alg-macaulay}
	For every prime power $q$ and $m \in \N$,
	there exists an algorithm
	which takes a set of \(n\) points \(V\) in \(\F_{q}^{m}\) as input
	and gives a minimal (w.r.t. size) Macaulay basis of \(\I(V)\) as output in the monomial basis in time
	\(\poly \left(n,\log q,\binom{m+d+1}{m}\right)\),
	where \(d\) is the extension degree of \(V\).
\end{lemma}
\begin{proof}
	Given a set of polynomials \(S\),
	define \(S_{i}\) to be the subset of \(S\)
	of polynomials of degree exactly \(i\),
	and define \(S_{\leq i}\)
	to be the subset of polynomials of degree $\leq i$.

	We will then construct the subsets of the minimal Macaulay basis
	\(\Grobner_{i}\) inductively.
	We first define \(\Grobner_{0} = \emptyset\),
	and then set \(\Grobner_{i}\)
	to be any basis of the quotient space
	\(\I(V)_{\leq i} / L_{i}\), where
	\begin{align*}
		L_{i} =
		\left\{\sum_{j} h_{j} \cdot g_{j}
		\; \middle| \;
		h_{j} \in \F_{q}[x_{1}, \dots, x_{m}],
		g_{j} \in \I(V)_{\leq i-1}, \deg(g_{j}h_{j}) \leq i\right\}.
	\end{align*}
	We repeat this step until \(i = d+1\),
	so in each step, we check whether or not the
	the monomials of degree at most \(i\) span all functions on \(|V|\),
	to know when to stop,
	as expressed in the following algorithm:
	%	A polynomial in a minimal Macaulay basis
	%	can have at most degree \(d+1\)
	%	  where \(d\) is the extension degree,
	%	so we have to iterate the above step at most \(d+1\) times,
	%	to get the full minimal Macaulay basis.
	%	We can then, in each step, calculate whether or not
	\begin{algobox}
		\begin{algorithm}[H]
			\caption{Constructing Macaulay basis.}
			\KwIn{Subset \(V\) of \(\F_{q}^{m}\) of size \(n\)}
			\KwOut{A Macaulay basis \(\Grobner\) of \(\I(V)\)}

			\vspace{4mm}

			\(\Grobner, \Grobner_{0}, \mathcal{A}_{0}, \mathcal{B}_{0} \gets \emptyset\).

			\For{\(i = 1, \dots\)}{
				Compute the evaluation matrix \(E_i \in \F_q^{n\times \binom{m+i}{i}}\)
				of all monomials of degree at most \(i\) on the points in \(V\)

				Compute a basis \(\mathcal{A}_{i}\) of \(\ker(E_i)\subset \F_q^{\binom{m+i}{i}}\)

				Compute a basis \(\mathcal{B}_{i}\) of
				\(\mathcal{A}_{i-1} +\Span \left(x_{i} a:\ a\in \mathcal{A}_{i-1},\ 1 \leq i\leq m\right)\)

				Compute a minimal basis \(\Grobner_{i}\) so that
				\(\Span \left(\Grobner_i\right) +\Span \left(\mathcal{B}_{i}\right) =
				\Span \left(\mathcal{A}_{i}\right)\)

				\(\Grobner \gets \Grobner \cup \Grobner_{i}\)

				\If{
					\(\rank E_{i-1} = n\)
				}{
					\textbf{break}
				}
			}
			\Return \(\Grobner\)
		\end{algorithm}
	\end{algobox}

	\paragraph{Correctness}
	We will now show any set \(\Grobner\)
	is a minimal Macaulay basis
	if and only if
	\(\Grobner_{i}\) is a basis for the quotient space
	\(\I(V)_{\leq i} / L_{i}\).
	To see this,
	if every \(\Grobner_{i}\) is a spanning set for
	\(\I(V)_{\leq i} / L_{i}\),
	then any polynomial \(P \in \I(V)_{\leq i}\)
	can be written as
	\begin{align*}
		P \equiv \sum_{g \in \Grobner_i} c_{g} \cdot g
		\mod
		L_{i},
	\end{align*}
	which is equivalent to
	\begin{align*}
		P =
		\sum_{g \in \Grobner_i} c_{g} \cdot g
		+
		\sum_{g \in \Grobner_{\leq i-1}} h_{g} \cdot g
	\end{align*}
	where \(c_{g}\) are constants and \(\deg(h_{g}\cdot g) \leq i\).
	This is the condition for \(\Grobner\) being a Macaulay basis,
	so \(\Grobner\) is a Macaulay basis if and only if
	each \(\Grobner_i\) spans \(\I(V)_{\leq i} / L_{i}\).

	Furthermore, note that in the above characterization,
	the conditions on each degree \(i\) is independent of each other,
	so \(\Grobner\) is minimal if and only if each \(\Grobner_{i}\)
	is minimal, which is equivalent to each of them being a basis.

	From the above, we see that a minimal Macaulay basis does not
	contain any degree \(i\) polynomials if
	\begin{align*}
		\I(V)_{\leq i} = L_{i}.
	\end{align*}
	We show this is true for any \(i \geq d +2\).
	Since any monomial \(m\) of degree \(i-1\),
	is equivalent to a polynomial \(P\) of degree \(i-2\), as \(i-2 \geq d\),
	it follows that \(x_{i}m - x_{i}P\) is in \(L_{i}\) for any
	variable \(x_{i}\). In particular, any polynomial of degree $i$ is equal to a polynomial of degree at most $i-1$ modulo $L_i.$
	Together with the fact that \(\I(V)_{i-1} \subseteq L_{i}\),
	we have \(\I(V)_{i} = L_{i}\).

	\paragraph{Runtime}
	To analyze the runtime,
	in each step of the loop
	we perform Gaussian elimination on matrices of size at most
	\(n \times \binom{m +d+1}{m}\),
	so each step takes at most time
	\(\poly(n,\log q,\binom{m+d+1}{m})\),
	and so the total runtime must also be
	\(\poly(n,\log q,\binom{m+d+1}{m})\).
\end{proof}

\begin{claim}
	\label{clm:rt-gauss}
	The running times in \Cref{lem:alg-extension-degree} and \Cref{lem:alg-macaulay} are both $q^{O(m)}.$
\end{claim}

\begin{proof}
	This follows from the fact that $n= |V|\leq q^m$, the extension degree $d$ is at most $q(m-1)$ (the extension degree of $\F_q^{m}$) and the following binomial estimate
	\[
		\binom{m+d+1}{m} \leq \binom{mq+1}{m} \leq \left(\frac{emq+e}{m}\right)^m \leq e^{m}(q+1)^m.
	\]
\end{proof}

\section{Zero-on-Variety Test}\label{sec:vanishing-certificate}
In this subsection, we discuss an efficient test to decide whether a given oracle vanishes on a subset of points/variety.
Let $V \subset \F_{q}^{m}$ be a set with a Macaulay basis \(\Grobner(V)\) with extension degree $d_{V}$ and complexity $k$ (see \Cref{sec:generatingset} for formal definitions).
Let $P:\F_{q}[x_{1}, \ldots, x_{m}] \to \F_{q}$ be a polynomial of degree $d$ and say $d \leq d_{V}$. Informally, the main goal of this section is:
\begin{center}
	\textit{Design an efficient standard verifier to decide whether $P$ is zero at all points of $V$.}
\end{center}

\paragraph*{}Before we state our standard verifier, let us first discuss what could constitute as proof for the vanishing of $P|_{V}$.
Suppose $P(\mathbf{x})$ is a degree $d$ polynomial and $P|_{V} \equiv 0$.
Using the definition of the ideal $\mathbb{I}(V)$ and $\Grobner(V)$, we have:\newline
\begin{gather*}
	P|_{V} \equiv 0 \quad \iff \quad P \in \mathbb{I}(V) \\
	\iff \text{There exists polynomials } \; h_{g}\in \F_{q}[x_{1},\ldots,x_{m}] \; \text{ for every } \; g \in \Grobner(V) \; \text{ such that}
\end{gather*}
\begin{equation}\label{eqn:vanishing-certificate}
	P(\mathbf{x}) \; = \; \sum_{g \in \mathfrak{G}(V)} \; h_{g}(\mathbf{x}) \cdot g(\mathbf{x}), \quad \text{ where for every $g$, } \, \deg(h_{g}(\mathbf{x}) \cdot g(\mathbf{x})) \leq d.
\end{equation}
We will refer to the ordered tuple $(h_{g}:g\in \mathfrak{G}(V)) \in (\F_{q}[x_{1},\ldots,x_{m}])^{k}$ in \Cref{eqn:vanishing-certificate}
as a \emph{vanishing certificate}\footnote{There could be multiple vanishing certificates for $P|_{V}$ satisfying \Cref{eqn:vanishing-certificate}.
We only use the fact that there always exists a vanishing certificate where \emph{each} polynomial $h_{g}$ has degree $\leq \deg(P)$.
We are guaranteed of the existence of such a vanishing certificate due to $\Grobner(V)$.} for the polynomial $P|_{V}$.\\

\noindent
\begin{definition}[Vanishing Certificate Polynomial]\label{defn:vanishing-cert-poly}
	Let $P$, $V \subseteq \F_{q}^{m}$ and $\Grobner(V)$ as defined above. A \emph{vanishing certificate polynomial} $\mathcal{M}_{V}(P): \F_{q}^{m+k} \to \F_{q}$ is a polynomial of degree $\leq d$ satisfying the following conditions:
	\begin{itemize}
		\item There exists polynomials $h_{g} \in \F_{q}[x_{1},\ldots,x_{m}]$ for every $g \in \mathfrak{G}(V)$ such that
		      \begin{align*}
			      \mathcal{M}_{V}(P)(\mathbf{x},\mathbf{y}) \; = \; \sum_{g \in \Grobner(V)} \; h_{g}(\mathbf{x}) \cdot y_{g}.
		      \end{align*}
		\item If we substitute $g(\mathbf{x})$ for $y_{g}$ for every $g \in \Grobner(V)$ in the polynomial $\mathcal{M}_{V}(P)(\mathbf{x},\mathbf{y})$, it should be the polynomial $P(\mathbf{x})$, i.e.
		      \begin{align*}
			      P(\mathbf{x}) \; = \;  \mathcal{M}_{V}(P)(\mathbf{x}, \; (g(\mathbf{x}) \, : \, g \in \Grobner(V) )  ),
		      \end{align*}
		      where the above equality is equality as polynomials.
	\end{itemize}
\end{definition}

\noindent
Whenever the subset $V \subseteq \F_{q}^{m}$ is clear from the context, we will use $\mathcal{M}_{P,\mathrm{lines}}^{(d)}$ to refer to the $d^{th}$ lines table for $\mathcal{M}_{V}(P)$ (see \Cref{defn:lines-table} for a formal definition of the lines table).\\
We record our discussion above using \Cref{defn:vanishing-cert-poly} in the following observation. We use the same notation as in \Cref{defn:vanishing-cert-poly}.\\

\begin{observation}\label{obs:valid-proof-vanishing}
	Let $P(\mathbf{x})$ be a degree $d$ polynomial. Then $P|_{V} \equiv 0$ if and only if there exists a vanishing certificate polynomial $\mathcal{M}_{V}(P)$ of degree $\leq d$ (see \Cref{defn:vanishing-cert-poly}). We would like to emphasize that \Cref{defn:vanishing-cert-poly} has a degree restriction on $\mathcal{M}_{V}(P)$ and this will be crucial for us, as we will see soon.
\end{observation}

\noindent
\Cref{obs:valid-proof-vanishing} says that if a verifier wants to test whether $P$ vanishes on $V$, a valid proof $\Pi$ is:
\begin{equation}\label{eqn:correct-proof-vanishing-test}
	\Pi \; = \; \paren{\mathcal{M}_{V}(P), \; \mathcal{M}_{P,\mathrm{lines}}^{(d)}}.
\end{equation}
In other words, our verifier for the Zero-on-Variety test will accept the above $\Pi$ with probability $1$. For a degree $d$ polynomial $P$ where $P|_{V} \notequiv 0$, no vanishing certificate polynomial exists, and we require our verifier to reject every ``claimed'' proof $\Pi'$ with high probability.

We next describe a standard verifier $\vvanish$ (recall the definition of standard verifier in \Cref{defn:standard-verifier}) with oracle access to a function $f$ and an arbitrary string $\Pi$ to decide whether $f|_{V} \equiv 0$, in \Cref{algo:vanishing-grid}. For convenience in writing, we define the following map:
\begin{gather*}
	\varphi: \F_{q}^{m} \; \to \; \F_{q}^{k} \\
	(z_{1},\ldots,z_{m}) \; \mapsto \; (g(\mathbf{z}) : g \in \mathfrak{G}(V))
\end{gather*}

\begin{algobox}
	\begin{algorithm}[H]
		\caption{Zero-on-Variety Test for $V$: $\vvanish^{(\cdot)}_{\cdot}$}
		\label{algo:vanishing-grid}

		\DontPrintSemicolon

		\KwIn{Degree parameter $d$, subset $V \subseteq \F_{q}^{m}$, Macaulay basis $\Grobner(V)$,\newline
		randomness $\mathbf{a},\mathbf{b} \in \F_{q}^{m+k}$, $\bm{\alpha} \in \F_{q}^{m}$, element $t \in \F_{q}^{\times}$, and \newline oracle access to $(f,\mathcal{M},\mathcal{M}')$ where $f \in (\F_{q})^{q^{m}}$, $\mathcal{M} \in (\F_{q})^{q^{m+k}}$, and $\mathcal{M}' \in (\F_{q}^{d})^{q^{2(m+k)}}$.}

		\vspace{3mm}
		Run $\vldt^{\mathcal{M},\mathcal{M}'}_{d}(;\mathbf{a},\mathbf{b},t)$ (see \Cref{algo:low-degree-test}) \tcp*{Two queries to $(\mathcal{M},\mathcal{M}')$}

		\vspace{3mm}

		\If{$\vldt^{\mathcal{M},\mathcal{M}'}_{d}(;\mathbf{a},\mathbf{b},t)$ returns \texttt{Reject}}{
			\vspace{2.5mm}
			\Return{\texttt{Reject}}
		}

		\vspace{2mm}

		Run $\corr_{d}^{\mathcal{M},\mathcal{M}'}((\bm{\alpha}, \mathbf{0}); \mathbf{a},t)$ (see \Cref{algo:local-correction}) \tcp*{Two queries to $(\mathcal{M},\mathcal{M}')$}

		\vspace{3mm}

		\If{$\corr_{d}^{\mathcal{M,\mathcal{M}'}}((\bm{\alpha}, \mathbf{0}); \mathbf{a},t) \neq 0$}{
		\vspace{2mm}
		\Return{\texttt{Reject}}
		}

		\vspace{2mm}

		Run $\corr_{d}^{\mathcal{M},\mathcal{M}'}((\bm{\alpha}, \varphi(\bm{\alpha})) ; \mathbf{a},t)$ \tcp*{Two queries to $(\mathcal{M},\mathcal{M}')$ and time to evaluate $\varphi(\bm{\alpha})$ is $\bigO(k \cdot q^{\bigO(m)})$}
		\vspace{2mm}
		Query $f[\bm{\alpha}]$ \tcp*{One query to $f$}

		\vspace{3mm}

		\If{$\corr_{d}^{\mathcal{M},\mathcal{M}'}((\bm{\alpha}, \varphi(\bm{\alpha})) ; \mathbf{a},t) \neq f[\bm{\alpha}]$}{
		\vspace{2.5mm}
		\Return{\texttt{Reject}}
		\Else{
			\Return{\texttt{Accept}}
		}
		}

	\end{algorithm}
\end{algobox}

\begin{thmbox}
	\begin{lemma}[Zero-on-Variety Test]\label{lemma:vanishing-test}
		There exists an absolute constant $C > 0$ such that for every $d,q \in \mathbb{N}$ satisfying $q > Cd^{3}$,
		for every subset $V \subset \F_{q}^{m}$ with extension degree $\leq d$ and Macaulay complexity $k$, the following holds.
		The standard verifier $\vvanish$ satisfies the following properties:\\

		\noindent
		Let $r = (\mathbf{a},\mathbf{b},t,\bm{\alpha})$. Then,\\

		\begin{enumerate}
			\item \textbf{Completeness}: For every degree $d$ polynomial $f : \F_{q}^{m} \to \F_{q}$ with $f|_{V} \equiv 0$, there exists a proof $\Pi$ over alphabet $\F_{q}^{d+1}$ and size $\bigO(q^{2(m+k)})$ such that the following holds:
			      \begin{align*}
				      \Pr_{r}[\vvanish^{(f,\Pi)}_{d}(;r) \, \text{ returns } \, \texttt{Accept}] \; = \; 1.
			      \end{align*}

			\item \textbf{Soundness}: Let $f : \F_{q}^{m} \to \F_{q}$ be any function for which there exists a unique
			      degree $d$ polynomial $P(\mathbf{x})$ such that $\delta(f,P) = \delta < 0.01$ and $P|_{V} \notequiv 0$. Then for every string $\Pi$, the following holds:
			      \begin{align*}
				      \Pr_{r}[\vvanish^{(f,\Pi)}_{d}(;r) \, \text{ returns } \, \texttt{Reject}] \; \ge \; 0.04.
			      \end{align*}
			\item \textbf{Efficiency}: $\vvanish$ uses $\bigO((m + k) \log q)$ bits of randomness, makes $7$ oracle queries to $(f,\Pi)$, and runs in time $\bigO(k \cdot q^{\bigO(m)})$.
		\end{enumerate}
	\end{lemma}
\end{thmbox}
\begin{proof}[Proof of \Cref{lemma:vanishing-test}]
	We first note that the efficiency immediately follows from the comments in \Cref{algo:vanishing-grid}. We discuss completeness next.

	\paragraph*{Completeness}Suppose $f$ is a polynomial of degree at most $d$ and $f|_{V} \equiv 0$. As observed in \Cref{obs:valid-proof-vanishing}, we know there exists a vanishing certificate polynomial $\mathcal{M}_{V}(f)$ of degree $d$, and let $\Pi$ be as stated in \Cref{eqn:correct-proof-vanishing-test}. From the first item of \Cref{thm:low-degree-testing}, we know that $\vldt^{\mathcal{M}_{V}(f),\mathcal{M}_{f,\mathrm{lines}}^{(d)}}$ returns \texttt{Accept} with probability $1$. From the first item of \Cref{thm:local-correction}, we know that $\corr^{\mathcal{M}_{V}(f),\mathcal{M}_{f,\mathrm{lines}}^{(d)}}(\bm{\alpha},\mathbf{0})$ is equal to $0$ for every $\bm{\alpha} \in \F_{q}^{m}$ with probability $1$. Similarly, we know that $\corr^{\mathcal{M}_{V}(f),\mathcal{M}_{f,\mathrm{lines}}^{(d)}}(\bm{\alpha},\varphi(\bm{\alpha}))$ is equal to $f(\bm{\alpha})$ for every $\bm{\alpha} \in \F_{q}^{m}$ with probability $1$.
	It is not difficult to see that \Cref{algo:vanishing-grid} accepts $\Pi$ with probability $1$. This finishes the completeness part of \Cref{lemma:vanishing-test}.

	\paragraph*{Soundness}Let $f$ be any function for which there exists a degree $d$ polynomial $P(\mathbf{x})$ such that $\delta(f,P) \leq \delta$ and $P|_{V} \notequiv 0$. Consider the following events from \Cref{algo:vanishing-grid}:
	\begin{enumerate}
		\item $\mathcal{E}_{1}$ denotes the event that $\vldt^{\mathcal{M},\mathcal{M}'}_{d}(;\mathbf{a},\mathbf{b},t)$ returns \texttt{Reject}. It depends on the choice of $(\mathbf{a},\mathbf{b},t)$.
		\item $\mathcal{E}_{2}$ denotes the event that $ \corr_{d}^{\mathcal{M},\mathcal{M}'}((\bm{\alpha}, \mathbf{0}); \mathbf{a},t) \neq 0$. It depends on the choice of $(\bm{\alpha}, \mathbf{a},t)$.
		\item $\mathcal{E}_{3}$ denotes the event that $\corr_{d}^{\mathcal{M},\mathcal{M}'}((\bm{\alpha}, \varphi(\bm{\alpha})); \mathbf{a},t) \neq f[\bm{\alpha}]$. It depends on the choice of $(\bm{\alpha},\mathbf{a},t)$.
	\end{enumerate}
	In the proof below, to avoid cumbersome writing, we will avoid repeatedly mentioning the random bits that each event depends on.

	\paragraph*{}If either of the events $\mathcal{E}_{1}$ or $\mathcal{E}_{2}$ happens with probability greater than $ 0.04$, then we have the desired soundness. Assume that is not the case, i.e.,
	\begin{align*}
		\Pr_{\mathbf{a},\mathbf{b},t}[\mathcal{E}_{1}] \; \leq \; 0.04 \quad \text{ and } \quad \Pr_{\bm{\alpha},\mathbf{a},t}[\mathcal{E}_{2}] \; \leq \; 0.04.
	\end{align*}
	We now want to argue that $\mathcal{E}_{3}$ happens with probability at least $0.04$.\\

	\noindent
	Since $\mathcal{E}_{1}$ happens with probability at most $0.04$,
	from \Cref{thm:low-degree-testing},
	we know that there exists a degree $d$ polynomial $\mathcal{R}(\mathbf{x},\mathbf{y})$ such that
	$\delta(\mathcal{M}, \mathcal{R}) \; \leq \; 0.16$.
	We will show the following claim.\\

	\begin{claim}\label{claim:vanishing-test-soundness}
		Let $\mathcal{R}(\mathbf{x},\mathbf{y}): \F_{q}^{m+k} \to \F_{q}$ be the degree $d$ polynomial such that $\delta(\mathcal{M},\mathcal{R}) \leq 0.16$. Then,
		\begin{align*}
			\mathcal{R}(\mathbf{x},\mathbf{0}) \equiv 0.
		\end{align*}
	\end{claim}
	\begin{proof}[Proof of \Cref{claim:vanishing-test-soundness}]
		As mentioned above, we know that $\delta(\mathcal{M},\mathcal{R}) \leq 0.16$. Using \Cref{thm:local-correction}, we get:
		\begin{align*}
			\Pr_{\mathbf{a},t}[\corr^{\mathcal{M},\mathcal{M}'}_{d}((\bm{\alpha},\mathbf{0}); \mathbf{a},t) = \mathcal{R}(\bm{\alpha},\mathbf{0})] \; \geq \; 1 - 2\sqrt{0.16} - \dfrac{d}{q-1}, &  & \text{for every } \bm{\alpha} \in \F_{q}^{m}
		\end{align*}
		\begin{equation}\label{eqn:vanish-soundness-1}
			\Rightarrow \Pr_{\bm{\alpha}, \mathbf{a},t}[\corr^{\mathcal{M},\mathcal{M}'}_{d}((\bm{\alpha},\mathbf{0}); \mathbf{a},t) \neq \mathcal{R}(\bm{\alpha},\mathbf{0})] \; \leq \; 0.8 + \dfrac{d}{q-1}.
		\end{equation}
		Since the event $\mathcal{E}_{2}$ happens with probability $\leq 0.04$, we have,
		\begin{equation}\label{eqn:vanish-soundness-2}
			\Pr_{\bm{\alpha},\mathbf{a},t}[\corr^{\mathcal{M},\mathcal{M}'}_{d}((\bm{\alpha},\mathbf{0}); \mathbf{a},t) \neq 0] \; \leq \; 0.04.
		\end{equation}
		Using union bound on \Cref{eqn:vanish-soundness-1} and \Cref{eqn:vanish-soundness-2}, we get,
		\begin{align*}
			\Pr_{\bm{\alpha},\mathbf{a},t}[\mathcal{R}(\bm{\alpha},\mathbf{0}) \neq 0] \; \leq \; 0.84  + \dfrac{d}{q-1} \\ \\
			\iff \quad \Pr_{\bm{\alpha},\mathbf{a},t}[\mathcal{R}(\bm{\alpha},\mathbf{0}) = 0] \; \geq \; 0.16 - \dfrac{d}{q-1}.
		\end{align*}
		Since the event $(\mathcal{R}(\bm{\alpha},\mathbf{0}) = 0)$ does not depend on the random choice of $(\mathbf{a},t)$, we get,
		\begin{align*}
			\Pr_{\bm{\alpha}}[\mathcal{R}(\bm{\alpha},\mathbf{0}) = 0] \; \geq \; 0.16 - \dfrac{d}{q-1}.
		\end{align*}
		We choose $C$ in the statement of \Cref{lemma:vanishing-test} large enough such that $0.16 - \frac{d}{q-1} > \frac{d}{q}$. The polynomial distance lemma (\Cref{thm:odlsz}) then implies that $\mathcal{R}(\mathbf{x},\mathbf{0}) \equiv 0$. This finishes the proof of \Cref{claim:vanishing-test-soundness}.
	\end{proof}

	\paragraph*{}\Cref{claim:vanishing-test-soundness} implies that $\mathcal{R}(\mathbf{x},\mathbf{y})$ belongs to the ideal $\mathbb{I}(y_{1},\ldots,y_{k})$. This implies the existence of polynomials $R_{1},\ldots,R_{k} \in \F_{q}[\mathbf{x},\mathbf{y}]$ such that  $\mathcal{R}(\mathbf{x},\mathbf{y})$ can be expressed as follows:
	\begin{align*}
		\mathcal{R}(\mathbf{x},\mathbf{y}) \; = \; \sum_{g \in \Grobner(H)} \; R_{g}(\mathbf{x},\mathbf{y}) \cdot y_{g}.
	\end{align*}
	Define the polynomial $R:\F_{q}^{m} \to \F_{q}$ as follow, $R(\mathbf{x}) \; := \; \mathcal{R}(\mathbf{x}, \varphi(\mathbf{x}))$. Observe that $P(\mathbf{x})$ and $R(\mathbf{x})$ are distinct polynomials, otherwise $P|_{V} \equiv 0$.\newline

	\noindent
	Since for every $g \in \mathfrak{G}(V)$, we know that $\deg(g)  \leq d$ and we also have that $\deg(\mathcal{R}(\mathbf{x},\mathbf{y})) \leq d$, we get $\deg(R(\mathbf{x})) \leq d^{2}$. Since $\delta(\mathcal{M},\mathcal{R}) \leq 0.16$, we have
	\begin{align*}
		\Pr_{\mathbf{a},t}[\corr^{\mathcal{M},\mathcal{M}'}_{d}((\bm{\alpha},\varphi(\bm{\alpha})); \mathbf{a},t) = \mathcal{R}(\bm{\alpha},\varphi(\bm{\alpha}))] \; \geq \;
		1 - 2\sqrt{0.16} - \dfrac{d}{q-1}, &  & \text{for every} \; \bm{\alpha} \in \F_{q}^{m} \\ \\
		\Rightarrow  \Pr_{\bm{\alpha},\mathbf{a},t}[\corr^{\mathcal{M},\mathcal{M}'}_{d}((\bm{\alpha},\varphi(\bm{\alpha})); \mathbf{a},t) = \mathcal{R}(\bm{\alpha}, \varphi(\bm{\alpha}))] \; \geq \;
		0.2 - \dfrac{d}{q-1}
	\end{align*}
	\begin{equation}\label{eqn:vanish-soundness-3}
		\Rightarrow \Pr_{\bm{\alpha},\mathbf{a},t}[\corr^{\mathcal{M},\mathcal{M}'}_{d}((\bm{\alpha},\varphi(\bm{\alpha})); \mathbf{a},t) \neq R(\bm{\alpha})] \; \leq \;
		0.8 + \dfrac{d}{q-1}.
	\end{equation}
	Recall that $P \neq R$ and from the polynomial distance lemma (\Cref{thm:odlsz}), we have:
	\begin{equation}\label{eqn:vanish-soundness-4}
		\Pr_{\bm{\alpha}}[R(\bm{\alpha}) = P(\bm{\alpha})] \; \leq \; \dfrac{d^{2}}{q} \quad \Longrightarrow \quad \Pr_{\bm{\alpha}}[R(\bm{\alpha}) = f[\bm{\alpha}]] \; \leq \; \delta + \dfrac{d^{2}}{q}.
	\end{equation}
	Using \Cref{eqn:vanish-soundness-3} and \Cref{eqn:vanish-soundness-4} and applying union bound, we get,
	\begin{gather*}
		\Pr_{\bm{\alpha},\mathbf{a},t}[\corr^{\mathcal{M},\mathcal{M}'}_{d}((\bm{\alpha},\varphi(\bm{\alpha})); \mathbf{a},t) = f[\bm{\alpha}]]
		\; \\ \\
		\leq \; \Pr_{\bm{\alpha}, \mathbf{a},t}[\corr^{\mathcal{M},\mathcal{M}'}_{d}((\bm{\alpha},\varphi(\bm{\alpha})); \mathbf{a},t) \neq R(\bm{\alpha})] \, + \, \Pr_{\bm{\alpha}}[R(\bm{\alpha}) = f[\bm{\alpha}]] \quad
		\leq \quad	0.8 + \dfrac{2d^{2}}{q-1} + \delta.
	\end{gather*}
	By choosing $C$ appropriately in the statement of \Cref{lemma:vanishing-test}, we can set $2d^{2}/(q-1) \leq 0.01$. Thus,
	\begin{align*}
		\Pr_{\bm{\alpha},\mathbf{a},t}[\corr^{\mathcal{M},\mathcal{M}'}_{d^{2}}((\bm{\alpha}, \varphi(\bm{\alpha})); \mathbf{a},t) \neq f[\bm{\alpha}]] \; \geq \; 0.19 - \delta \; \geq \; 0.04,
	\end{align*} where we are using $\delta \leq 0.01$.
	Thus $\mathcal{E}_{3}$ happens with probability $\geq 0.04$. Hence we have showed that either $\mathcal{E}_{1}$ or $\mathcal{E}_{2}$ happens with probability $\geq 0.04$, otherwise $\mathcal{E}_{3}$ happens with probability $\geq 0.04$. This finishes the soundness of \Cref{lemma:vanishing-test} and also the proof of \Cref{lemma:vanishing-test}.
\end{proof}

\section{Proof of the Main Theorem (\Cref{thm:main-thm})}
\label{sec:main-thm}

In this section, we give the proof of \Cref{thm:main-thm}, which we recall below.\\

\mainthm*

\begin{proof}[Proof of \Cref{thm:main-thm}]
	We will consider the $\mathsf{NP}$-complete problem $3$-$\mathsf{COLOR}$ for graphs (see \Cref{defn:3-color} for a formal definition of $3$-$\mathsf{COLOR}$). We will use $\mathcal{V}_{\mathsf{PCP}}$ to denote the standard verifier with parameters as stated in \Cref{thm:main-thm}.\newline

\paragraph*{}From \Cref{lem:alg-macaulay} and \Cref{clm:rt-gauss}, we know that $\mathcal{V}_{\mathsf{PCP}}$ can compute the Macaulay basis of complexity $k$ in time $q^{\bigO(m)}$. Let $E: V \times V \to \Boo \subset \F_{q}$	be the edge function for the input graph $G = (V,E)$, defined as $E(u,v) = 1$ if and only if $(u,v) \in E$. Let $\widehat{E}: \F_{q}^{m} \times \F_{q}^{m} \to \F_{q}$ denote an extension of $E$ and from the second item of \Cref{lemma:subaddgrob}, we know that $\deg(\widehat{E}) \leq 2d$. By \Cref{lem:alg-extension-degree} and \Cref{clm:rt-gauss}, the standard verifier $\mathcal{V}_{\mathsf{PCP}}$ can compute both, the extension degree $d$ and the extension $\widehat{E}$ in time $q^{\bigO(m)}$.

\paragraph*{}For simplicity, we will describe a standard verifier $\mathcal{V}$ for $3$-$\mathsf{COLOR}$ and then $\mathcal{V}_{\mathsf{PCP}}$ will be repeating $\mathcal{V}$ for $\bigO(1)$ times. More particularly, the standard verifier $\mathcal{V}$ will have soundness $\gamma$ for some absolute constant $\gamma \in (0,1)$, i.e. $\mathcal{V}$ rejects with probability at least $\gamma$. The standard verifier $\mathcal{V}_{\mathsf{PCP}}$ will repeat $\mathcal{V}$ for $\bigO(1/\gamma)$ times and return \texttt{Reject} if any one of the iterations return \texttt{Reject}. As $\bigO(1/\gamma) = \bigO(1)$, the number of random bits, queries, and running time of $\mathcal{V}_{\mathsf{PCP}}$ are a constant factor multiple of the number of random bits, queries, and running time of $\mathcal{V}$ respectively. So for rest of the proof, it will be sufficient to describe a standard verifier $\mathcal{V}$ which uses $c'(k+m) \log q$ random bits, makes $\ell'$ queries to proofs over alphabets of size $c' \cdot (d \log q)$, have soundness guarantee of $\gamma$, and has running time $q^{\bigO(m+k)}$, for some constants $c',\ell',$ and $\gamma \in (0,1)$. From the previous paragraph, we know that $\mathcal{V}_{\mathsf{PCP}}$ can compute the Macaulay basis, extension degree $d$, and extension $\widehat{E}$, all in time $q^{\bigO(m)}$. So we will assume that our standard verifier $\mathcal{V}$ has access to all of them.

	\paragraph*{Oracles}
	We now describe oracles that the standard verifier $\mathcal{V}$ expects in a proof $\Pi$. In particular, if $G \in 3$-$\mathsf{COLOR}$, then our standard verifier always returns \texttt{Accept}. Our oracles will be evaluation tables of polynomials and their corresponding lines table. In the following, we also mention the size of each oracle that appears in the proof.

	\begin{enumerate}

		\item Let $\chi: V \to \set{-1,0,1} \subset \F_{q}$ be a coloring assignment to every vertex in the input $G = (V,E)$.
		      Here we use $\set{-1,0,1}$ to denote three distinct colors.\\

		      \noindent
		      Let $\widehat{\chi}: \F_{q}^{m} \to \F_{q}$ denote an extension of $\chi$ of degree $d$. Let $\widehat{\chi}_{\mathrm{lines}}^{(d)}$ be the $d^{th}$ lines table for $\widehat{\chi}$.\newline
		      \vspace{2mm}
		      \textcolor{magenta}{\texttt{Size of $(\widehat{\chi}, \widehat{\chi}_{\mathrm{lines}}^{(d)})$ is $ 2^{\bigO(m \log q)}$ over alphabet of size $\bigO(d \log q)$.}}

		\item Define the polynomial $A: \F_{q}^{m} \to \F_{q}$ as follows:
		      \begin{align*}
			      A(\mathbf{x}) \; := \; \widehat{\chi}(\mathbf{x}) \cdot (\widehat{\chi}(\mathbf{x}) - 1) \cdot (\widehat{\chi}(\mathbf{x}) + 1).
		      \end{align*}
		      We have $\deg(A) \leq 3 \cdot \deg(\widehat{\chi}) \leq 3d$. Let $A_{\mathrm{lines}}^{(3d)}$ be the $(3d)^{th}$ lines table for $A$.\newline
		      \vspace{2mm}
		      \textcolor{magenta}{\texttt{Size of $(A,A_{\mathrm{lines}}^{(3d)})$ is $2^{\bigO(m \log q)}$ over alphabet of size $\bigO(d \log q)$.\\}}

		      \begin{observation}\label{obs:A-vanishes}
			      For a vertex $\mathbf{u} \in V$, $A(\mathbf{u}) = 0$ if and only if $\widehat{\chi}(\mathbf{u}) \in \set{-1,0,1}$. This implies that $A|_{V} \equiv 0$ if and only if for every vertex $\mathbf{u} \in V$, we have $\widehat{\chi}(\mathbf{u}) \in \set{-1,0,1}$.\\
		      \end{observation}

		      \noindent
		      Let\footnote{Recall $k$ is the Macaulay complexity of $\mathfrak{G}(V)$.} $\mathcal{M}_{A}: \F_{q}^{m+k} \to \F_{q}$ denote a vanishing certificate polynomial for $A|_{V}$ (see \Cref{defn:vanishing-cert-poly} for a formal definition). Let $d_{1} := \deg(\mathcal{M}_{A}) \leq \deg(A) \leq 3d$. Let $\mathcal{M}_{A,\mathrm{lines}}^{(3d)}$ be the $(3d)^{th}$ lines table for $\mathcal{M}_{A}$.\newline
		      \vspace{2mm}
		      \textcolor{magenta}{\texttt{Size of $(\mathcal{M}_{A}, \mathcal{M}_{A,\mathrm{lines}}^{(3d)})$ is $2^{\bigO(m\log q + k \log q)}$ over alphabet of size $\bigO(d \log q)$.}}

		\item Define the polynomial $B: \F_{q}^{m} \times \F_{q}^{m} \to \F_{q}$ as follows:
		      \begin{align*}
			      B(\mathbf{x}, \mathbf{y}) \; := \; \widehat{E}(\mathbf{x},\mathbf{y}) \cdot \prod_{a \in \set{\pm 1, \pm 2}} (\widehat{\chi}(\mathbf{x}) - \widehat{\chi}(\mathbf{y}) - a) \cdot
		      \end{align*}
		      We have $\deg(B) \leq \deg(\widehat{E}) + 4 \deg(\widehat{\chi}) \leq 6 d$. Let $B_{\mathrm{lines}}^{(6d)}$ be the $(6d)^{th}$ lines table for $B$.\newline
		      \vspace{2mm}
		      \textcolor{magenta}{\texttt{Size of $(B,B_{\mathrm{lines}}^{(6d)})$ is $2^{\bigO(m \log q)}$ over alphabet of size $\bigO(d \log q)$.\\}}

		      \begin{observation}\label{obs:B-vanishes}
			      Suppose $\widehat{\chi}(\mathbf{u}) \in \set{-1,0,1}$ for every $\mathbf{u} \in V$. For any two vertices $\mathbf{u}$ and $\mathbf{v}$, $B(\mathbf{u},\mathbf{v}) = 0$ if and only if either $(\mathbf{u},\mathbf{v}) \notin E$ or $\widehat{\chi}(\mathbf{u}) \neq \widehat{\chi}(\mathbf{v})$.\\

		      \end{observation}

		      \noindent
		      Let $\mathcal{M}_{B}$ denote a vanishing certificate polynomial for $B|_{V}$. Let $d_{2} := \deg(\mathcal{M}_{B}) \leq \deg(B)  \leq 6d$. Let $\mathcal{M}_{B,\mathrm{lines}}^{(6d)}$ be the $(6d)^{th}$ lines table for $\mathcal{M}_{B}$.\newline
		      \vspace{2mm}
		      \textcolor{magenta}{\texttt{Size of $(\mathcal{M}_{B}, \mathcal{M}_{B,\mathrm{lines}}^{(6d)})$ is $2^{\bigO(m\log q + k \log q)}$ over alphabet of size $\bigO(d \log q)$.}}

	\end{enumerate}

	The proof $\Pi$ consists of the following oracles:
	\begin{equation}\label{eqn:correct-proof}
		\Pi = \paren{\widehat{\chi}, \; \widehat{\chi}_{\mathrm{lines}}^{(d)}, A, \; A_{\mathrm{lines}}^{(3d)}, \; \mathcal{M}_{A}, \; \mathcal{M}_{A,\mathrm{lines}}^{(3d)}, \; B, \; B_{\mathrm{lines}}^{(6d)}, \; \mathcal{M}_{B}, \; \mathcal{M}_{B,\mathrm{lines}}^{(6d)}}
	\end{equation}
	As we have mentioned, the size of each of the components in $\Pi$, we get that the size of the proof $\Pi$ is $2^{\bigO(m \log q + k \log q)} = q^{\bigO(m+k)}$ over an alphabet of size $\bigO(d \log q)$.

	\paragraph*{Description of the standard verifier $\mathcal{V}$}We are now ready to describe the standard verifier $\mathcal{V}$ to test whether a graph $G$ is $3$-colorable or not. In the following description, we interpret that the proof $\Pi$ consists of the oracles as stated in \Cref{eqn:correct-proof}, i.e., $\mathcal{V}$ will interpret the proof $\Pi$ as a long string with sub-strings forming the structure in \Cref{eqn:correct-proof}. We will show that $\mathcal{V}$ is a standard verifier which uses $\bigO(m+k)$ random bits, makes $\bigO(1)$ queries to proofs over alphabets of size $\bigO(d \log q)$, have soundness guarantee of $\gamma$, and has running time $q^{\bigO(m+k)}$, for some constant  $\gamma \in (0,1)$.

	\begin{algobox}
		\begin{algorithm}[H]

			\caption{Test by the Verifier $\mathcal{V}$}
			\label{algo:verifier-tests}

			\DontPrintSemicolon

			\KwIn{Degree parameter $d$, subset $V$, Macaulay basis for $V$, polynomial $\widehat{E}$,\newline
            randomness $\mathbf{a},\mathbf{b} \in \F_{q}^{m}$, $\bm{\alpha}, \bm{\beta} \in \F_{q}^{2m}$, $\bm{\gamma}_{1},\bm{\gamma}_{2} \in \F_{q}^{m+k}$, $\bm{\mu}_{1},\bm{\mu}_{2} \in \F_{q}^{2(m+k)},$ $t \in \F_{q}^{\times}$, and 
				oracle access to $\Pi = (\widetilde{\chi}, \; \widetilde{\chi}', \; \widetilde{A}, \; \widetilde{A}', \mathcal{\widetilde{M}}_{A}, \mathcal{\widetilde{M}}_{A}', \; \widetilde{B}, \; \widetilde{B}', \mathcal{\widetilde{M}}_{B}, \mathcal{\widetilde{M}}_{B}')$}

			\vspace{4mm}

			Run $\vldt^{\widetilde{\chi}, \widetilde{\chi}'}_{d}(;\mathbf{a},\mathbf{b},t)$,   $\vldt^{\widetilde{A},\widetilde{A}'}_{3d}(;\mathbf{a},\mathbf{b},t)$, and $\vldt^{\widetilde{B},\widetilde{B}'}_{6d}(; \bm{\alpha},\bm{\beta},t)$ (see \Cref{algo:low-degree-test}) \tcp*{$6$ queries to $\Pi$ and runs in time $\poly(m,d)$}

			\vspace{4mm}

			\If{either of the above three $\vldt$ test returns \texttt{Reject}}{
				\Return{\texttt{Reject}}
			}

            \vspace{4mm}

			Query $\widetilde{\chi}[\mathbf{a}], \widetilde{\chi}[\mathbf{b}], \widetilde{A}[\mathbf{a}], \widetilde{B}[\mathbf{a},\mathbf{b}]$ \tcp*{$4$ queries to $\Pi$}
            
			\vspace{4mm}
			\If{$\widetilde{A}[\mathbf{a}] \neq \widetilde{\chi}[\mathbf{a}] \cdot (\widetilde{\chi}[\mathbf{a}]-1) \cdot (\widetilde{\chi}[\mathbf{a}]+1)$ OR $\widetilde{B}[\mathbf{a},\mathbf{b}] \neq \widehat{E}(\mathbf{a},\mathbf{b}) \cdot \prod_{i \in \set{\pm 1, \pm 2}} (\widetilde{\chi}[\mathbf{a}] - \widetilde{\chi}[\mathbf{b}] - i)$ \tcp*{Runs in time $\poly(n,q^{\bigO(m)})$}}{
				\Return{\texttt{Reject}}
			}

			\vspace{4mm}

			Run $\vvanish^{\widetilde{A}, \mathcal{\widetilde{M}}_{A}, \mathcal{\widetilde{M}}_{A}'}_{3d}(;\bm{\gamma}_{1}, \bm{\gamma}_{2}, \mathbf{a},t)$ and $\vvanish^{\widetilde{B}, \mathcal{\widetilde{M}}_{B}, \mathcal{\widetilde{M}}_{B}'}_{6d}(; \bm{\mu}_{1}, \bm{\mu}_{2}, \bm{\alpha}, t)$ (see \Cref{algo:vanishing-grid}) \tcp*{$14$ queries to $\Pi$ and running time $\poly(n,q^{\bigO(m+k)})$}

			\vspace{4mm}

			\If{either of the above two $\vvanish$ tests return \texttt{Reject}}{
				\Return{\texttt{Reject}}
			}

			\vspace{2mm}
			\Return{\texttt{Accept}}

		\end{algorithm}
	\end{algobox}

	\paragraph*{Efficiency}The random string used by $\mathcal{V}$ is the tuple $(\mathbf{a},\mathbf{b},\bm{\alpha},\bm{\beta},\bm{\gamma}_{1},\bm{\gamma}_{2}, \bm{\mu}_{1}, \bm{\mu}_{2},t)$. It is clear from here that these are $\bigO((m+k) \log q)$ random bits. From the comments in \Cref{algo:verifier-tests}, it is clear that $\mathcal{V}$ makes $\bigO(1)$ queries to the string $\Pi$, which is over an alphabet of size $\bigO(d \log q)$. From the comments in \Cref{algo:verifier-tests}, it is also clear that the running time of $\mathcal{V}$ is $q^{\bigO(m+k)}$.

	\paragraph*{Completeness}Let $G = (V,E) \in 3$-$\mathsf{COLOR}$. This means there exists a coloring $\chi: V \to \set{-1,0,1}$ such that for every edge $(\mathbf{u},\mathbf{v}) \in E$, we have $\chi(\mathbf{u}) \neq \chi(\mathbf{v})$. Let $\Pi$ be as stated in \Cref{eqn:correct-proof}. From the first item of \Cref{thm:low-degree-testing}, we know that all three low-degree tests $\vldt$ return \texttt{Accept} with probability $1$. From the definition of $\widehat{E}, \widehat{\chi}, A,$ and $B$, we know that $\mathcal{V}$ never returns \texttt{Reject} in Line 6 of \Cref{algo:verifier-tests}. Using \Cref{obs:A-vanishes} and \Cref{obs:B-vanishes}, we know that both $A|_{V} \equiv 0$ and $B|_{V \times V} \equiv 0$. From the completeness part of \Cref{lemma:vanishing-test}, we know that both $\vvanish^{A,\mathcal{M}_{A},\mathcal{M}_{A,\mathrm{lines}}^{(3d)}}$ and $\vvanish^{B,\mathcal{M}_{B},\mathcal{M}_{B,\mathrm{lines}}^{(6d)}}$ return \texttt{Accept} with probability $1$. Hence $\mathcal{V}^{\Pi}(G)$ always return \texttt{Accept} and thus has completeness $1$.

	\paragraph*{Soundness}Instead of the basic soundness claim, we prove a stronger claim that will also be useful in \Cref{sec:composition}.

	\begin{claim}[Soundness]
		\label{claim:robust-soundness-pcp}
		For any constant $\varepsilon > 0$ there is a constant $\gamma > 0$ such that the following holds. Suppose there exists no proper 3-coloring $\psi : V \to \set{-1,0,1}$ of $G$ such that $\delta(\widetilde{\chi}, \widehat{\psi}) \leq \varepsilon$, where $\widehat{\psi}$ is a degree $d$ extension of $\psi$. Then $\mathcal{V}$ returns \texttt{Reject} with probability at least $\gamma$. In particular, if $G$
		is not $3$-colorable, then $\mathcal{V}$ rejects with some constant probability.
	\end{claim}

	\begin{proof}[Proof of \Cref{claim:robust-soundness-pcp}]
		Consider the following events:
		\begin{itemize}
			\item $\mathcal{E}_{1}$ denotes the event that at least one of the three $\vldt$ test returns \texttt{Reject}. This event depends on the choice of $(\mathbf{a},\mathbf{b},\bm{\alpha},\bm{\beta},t)$.

			\item $\mathcal{E}_{2}$ denotes the event that
			      \begin{align*}
				      \widetilde{A}[\mathbf{a}] \neq \widetilde{\chi}[\mathbf{a}] \cdot (\widetilde{\chi}[\mathbf{a}]-1) \cdot (\widetilde{\chi}[\mathbf{a}]+1) \quad \text{ OR } \quad \widetilde{B}[\mathbf{a},\mathbf{b}] \neq \widehat{E}(\mathbf{a},\mathbf{b}) \cdot \prod_{i \in \set{\pm 1, \pm 2}} (\widetilde{\chi}[\mathbf{a}] - \widetilde{\chi}[\mathbf{b}] - i).
			      \end{align*}
			      This event depends on the choice of $(\mathbf{a},\mathbf{b})$.

			\item $\mathcal{E}_{3}$ denotes the event that $\vvanish^{\widetilde{A}, \mathcal{\widetilde{M}}_{A}, \mathcal{\widetilde{M}}_{A}'}_{3d}$ returns \texttt{Reject}. This event depends on the choice of $(\bm{\gamma}_{1},  \bm{\gamma}_{2}, \mathbf{a}, t)$.

			\item $\mathcal{E}_{4}$ denotes the event that $\vvanish^{\widetilde{B}, \mathcal{\widetilde{M}}_{B}, \mathcal{\widetilde{M}}_{B}'}_{6d}$ returns \texttt{Reject}. This event depends on the choice of $(\bm{\mu}_{1},\bm{\mu}_{2},\bm{\alpha},t)$.
		\end{itemize}

		\paragraph*{}Let $0 < \gamma < 0.01$ be an appropriate constant that we will choose later. If either of events $\mathcal{E}_{1}, \mathcal{E}_{2}$, or $\mathcal{E}_{3}$ happens with probability $> \gamma$, then we are done. Assume each of the events $\mathcal{E}_{1}, \mathcal{E}_{2}$, and $\mathcal{E}_{3}$ happens with probability $\leq \gamma$. We will show that $\mathcal{E}_{4}$ happens with probability $> \gamma$.

		\paragraph*{}Since $\mathcal{E}_{1}$ happens with probability $\leq \gamma$,
		\Cref{thm:low-degree-testing} implies:
		\begin{itemize}
			\item There exists degree $d$ polynomial $P_{\widetilde{\chi}}(\mathbf{x})$ such that $\delta(P_{\widetilde{\chi}}, \widetilde{\chi}) \leq 4 \gamma$.
			\item There exists degree $(3d)$ polynomial $P_{\widetilde{A}}(\mathbf{x})$ such that $\delta(P_{\widetilde{A}}, \widetilde{A}) \leq 4 \gamma$.
			\item There exists degree $(6d)$ polynomial $P_{\widetilde{B}}(\mathbf{x})$ such that $\delta(P_{\widetilde{B}}, \widetilde{B}) \leq 4 \gamma$.
		\end{itemize}

		\noindent
		We show the following claim on the relation between $P_{\widetilde{\chi}}, P_{\widetilde{A}}$, and $P_{\widetilde{B}}$.\\

		\begin{claim}\label{claim:relation-diff-P}
			Let the polynomials $P_{\widetilde{\chi}}, P_{\widetilde{A}}$, and $P_{\widetilde{B}}$ be as mentioned above. Then for every $\mathbf{x},\mathbf{y} \in \F_{q}^{m}$,
			\begin{equation}\label{eqn:P-chi-A-agree}
				P_{\widetilde{A}}(\mathbf{x}) \; = \; P_{\widetilde{\chi}}(\mathbf{x}) \cdot (P_{\widetilde{\chi}}(\mathbf{x})-1) \cdot (P_{\widetilde{\chi}}(\mathbf{x})+1).
			\end{equation}
			and
			\begin{equation}\label{eqn:P-chi-A-B-agree}
				P_{\widetilde{B}}(\mathbf{x},\mathbf{y}) \; = \; \widehat{E}(\mathbf{x},\mathbf{y}) \cdot \prod_{i \in \set{\pm 1, \pm 2}} (P_{\widetilde{\chi}}[\mathbf{x}] - P_{\widetilde{\chi}}[\mathbf{y}] - i).
			\end{equation}
		\end{claim}
		\begin{proof}[Proof of \Cref{claim:relation-diff-P}]
			The idea is to show that each pair of polynomials in either \Cref{eqn:P-chi-A-agree} or \Cref{eqn:P-chi-A-B-agree} agree on a large fraction of their respective domains. Since these are all low-degree polynomials, the polynomial distance lemma will imply that they are, in fact, equal.\\

			\noindent
			Since the event $\mathcal{E}_{2}$ happens with probability $\leq \gamma$, we have the following two inequalities:
			\begin{equation}\label{eqn:A-chi-agree}
				\Pr_{\mathbf{a} \sim \F_{q}^{m}}\left[\widetilde{A}[\mathbf{a}] \neq \widetilde{\chi}[\mathbf{a}] \cdot (\widetilde{\chi}[\mathbf{a}]-1) \cdot (\widetilde{\chi}[\mathbf{a}]+1)\right] \; \leq \; \gamma
			\end{equation}
			and also
			\begin{equation}\label{eqn:B-chi-agree}
				\Pr_{\mathbf{a},\mathbf{b} \sim \F_{q}^{m}}\left[\widetilde{B}[\mathbf{a},\mathbf{b}] \neq \widehat{E}(\mathbf{a},\mathbf{b}) \cdot \prod_{i \in \set{\pm 1, \pm 2}} (\widetilde{\chi}[\mathbf{a}] - \widetilde{\chi}[\mathbf{b}] - i)\right] \; \leq \; \gamma.
			\end{equation}
			Using $\delta(P_{\widetilde{\chi}}, \widetilde{\chi}) \leq 4 \gamma$, $\delta(P_{\widetilde{A}}, \widetilde{A}) \leq 4 \gamma$, and \Cref{eqn:A-chi-agree} together, via triangle inequality, we get:
			\begin{align*}
				\Pr_{\mathbf{a} \sim \F_{q}^{m}}\left[P_{\widetilde{A}}(\mathbf{a}) \neq P_{\widetilde{\chi}}(\mathbf{a}) \cdot (P_{\widetilde{\chi}}(\mathbf{a})-1) \cdot (P_{\widetilde{\chi}}(\mathbf{a})+1) \right] \; \leq \; 9 \gamma
			\end{align*}
			Since both the polynomials in the above inequality have degree $\leq (3d)$ and because $3d/q < 9 \gamma$, the polynomial distance lemma (\Cref{thm:odlsz}) implies \Cref{eqn:P-chi-A-agree}. An analogous argument shows \Cref{eqn:P-chi-A-B-agree}. This finishes the proof of \Cref{claim:relation-diff-P}.
		\end{proof}

		\paragraph*{}We assumed that event $\mathcal{E}_{3}$ happens with probability $\leq \gamma$. Recall that $\gamma < 0.01 < 0.04$. From \Cref{lemma:vanishing-test}, we can infer that $P|_{\widetilde{A}}$ vanishes on $V$. In other words, using \Cref{eqn:P-chi-A-agree}, we know that $P_{\widetilde{\chi}}$ is an extension of a valid vertex coloring $\widetilde{\chi}$ (i.e. $\widetilde{\chi}$ assigns each vertex a color from the set $\set{-1,0,1}$). We will now show that event $\mathcal{E}_{4}$ happens with probability $> \gamma$.\\

		\noindent
		Recall that $\widetilde{\chi}$ is $\varepsilon$-far from any degree $d$ extension of a proper $3$-coloring of $G$. In particular, for $\gamma < \varepsilon/4$, we know that $P_{\widetilde{\chi}}$ is not an extension of a proper $3$-coloring. In particular, there exists vertices $\mathbf{u}, \mathbf{v}$ such that $(\mathbf{u},\mathbf{v}) \in E$ and $P_{\widetilde{\chi}}(\mathbf{u}) = P_{\widetilde{\chi}}(\mathbf{v})$. In other words, $P_{\widetilde{B}}$ does not vanish on $V \times V$. From \Cref{lemma:vanishing-test}, we know that event $\mathcal{E}_{4}$ happens with probability $\geq 0.04 > \gamma$. Hence we have shown that either one of $\mathcal{E}_{1}, \mathcal{E}_{2}$, or $\mathcal{E}_{3}$ happens with probability $> \gamma$, otherwise $\mathcal{E}_{4}$ happens with probability $> \gamma$. This finishes the proof of \Cref{claim:robust-soundness-pcp}.

	\end{proof}
Hence we have shown that the standard verifier $\mathcal{V}$ has completeness $1$, soundness $\gamma$ for some constant $\gamma \in (0,1)$, uses $\bigO((m+k) \log q)$ random bits, makes $\bigO(1)$ queries to a proof of size $q^{\bigO(m+k)}$ over an alphabet of size $\bigO(d \log q)$, and runs in time $q^{\bigO(m+k)}$. As we discussed earlier, $\mathcal{V}_{\mathsf{PCP}}$ repeats $\mathcal{V}$ for $\bigO(1/\gamma) = \bigO(1)$ times to achieve soundness of $1/2$, and the other parameters remain the same upto $\bigO(1)$ factor.	This finishes the proof of \Cref{thm:main-thm}.
\end{proof}

\section{The PCP Theorem with One Composition}
\label{sec:composition}

In this section, we use the main theorem (\Cref{thm:main-thm}) to give a proof of the PCP theorem with a single composition, composing two different instantiations of our basic PCP.

More precisely, we need an extension of our PCP to a \emph{Robust Assignment-Tester} (or equivalently a \emph{Robust PCP of Proximity} \cite{BGHSV}). The definition below is due to Dinur and Reingold \cite{DinurReingold} with minor changes. In particular, our notion of robust soundness is defined for {\em every} distance parameter instead of fixed quantities. We will assume throughout this section that there is a single growing parameter $n$ and all other parameters ($r,\ell,a$) are (possibly constant) functions of $n.$

\paragraph{Notation.}Recall that two functions $f,g:S\rightarrow T$ are said to be \emph{$\delta$-close} if they differ on at most a $\delta$-fraction of their inputs and \emph{$\delta$-far} if they differ on at least a $\delta$-fraction of their inputs. For a predicate $A$ and an input $x$ in its domain, we define $\delta(x,A^{-1}(1))$ to be $\min_{y:A(y)=1} \delta(x,y)$.\\

\begin{definition}[Robust Assignment-Testers, combining Definitions 3.1 and 3.4 from \cite{DinurReingold}]
	\label{def:testers}
	A \emph{Robust Assignment-Tester} with parameters $(r,\ell,a,\rho)$ is a reduction whose input is a Boolean circuit $\varphi$ of size $n$ over Boolean variables $X$. Letting $R=R(n):=2^{r(n)}$, the reduction runs in time $\poly(n,R)$ and outputs a system of $R$ Boolean circuits $\psi = \{\psi_1, . . . , \psi_R\}$, each $\psi_{i}$ of size at most $a(n)$ over Boolean variables $X$ and auxiliary variables $Y$ such that the following conditions hold.
	\begin{itemize}
		\item Each $\psi_i$ depends on a set of $\ell$ variables from $X \cup Y$, which we denote $\mathrm{Vars}(\psi_i)$. The variables in Y take values in an
		      alphabet $\Sigma$, and are accessible to $\psi_i$ as a tuple of $ \lceil\log_2 (|\Sigma|)\rceil$ bits\footnote{In particular, this means that $\lceil \log_2(|\Sigma|)\rceil\le a(n)$.}.
		\item For every Boolean assignment $\sigma:X\rightarrow \{0,1\}$,
		      \begin{enumerate}
			      \item Completeness: If $\sigma$ satisfies the input circuit $\varphi$, then there exists an assignment $\tau:Y\rightarrow \Sigma$ such that $\sigma\cup \tau$ satisfies all of $\psi_1,\ldots, \psi_R$, i.e.:
                  $$\Pr_{i\in [R]}[\psi_i((\sigma \cup \tau)|_{\mathrm{Vars}(\psi_i)})=1]=1.$$
                  
			      \item {Robust Soundness}: For every assignment $\tau:Y\rightarrow \Sigma$, it holds that
                  $$\E_{i\in [R]}[\delta((\sigma \cup \tau)|_{\mathrm{Vars}(\psi_i)}, \psi_i^{-1}(1))] \ge \rho\cdot \delta(\sigma, \varphi^{-1}(1)).\footnote{Here the relative Hamming distance $\delta()$ is defined over bits since $\psi_i$ is defined to be a {\em Boolean} circuit.}$$
		      \end{enumerate}
	\end{itemize}
	A (not necessarily robust) Assignment-Tester with parameters $(r,\ell,a, \rho)$ is defined as above, except that we replace the Robust Soundness condition with the following weaker Soundness condition.
	\begin{itemize}
		\item Soundness: For every assignment $\tau:Y\rightarrow \Sigma$, it holds that
                  $$\Pr_{i\in [R]}[\psi_i((\sigma \cup \tau)|_{\mathrm{Vars}(\psi_i)})\ne 1] \ge \rho\cdot \delta(\sigma, \varphi^{-1}(1)).$$
	\end{itemize}
\end{definition}

\cite{DinurReingold} showed that the existence of one assignment-tester with specific parameters is enough to prove the PCP theorem. We state the observation formally below, adjusting to our definition.\\

\begin{observation}[Section 3.1 in~\cite{DinurReingold}]
	\label{obs:testersPCPs}
	To prove the PCP theorem i.e.,
	\begin{align*}
		\mathsf{NP}\subseteq \mathsf{PCP}_{1,1-\Omega(1)}[\bigO(\log n), \, \bigO(1), \, \bigO(1)],
	\end{align*}
	it suffices to show that there is an assignment-tester with parameters $r(n) = \bigO(\log n)$, $\ell=\bigO(1),a = \bigO(1)$ and $\rho = \Omega(1).$
\end{observation}

We will prove the existence of such an assignment-tester by composing two \emph{robust} assignment-testers with suitable parameters. The process of composition starts with two robust assignment-testers $\mathcal{A}_1$ and $\mathcal{A}_2$ and produces a robust assignment-tester $\mathcal{A}_3$ with parameters that are dictated by the parameters of $\mathcal{A}_1$ and $\mathcal{A}_2$. The robust assignment-tester $\mathcal{A}_{3}$ tests assignments to the Boolean circuit $\varphi$ that is an input to $\mathcal{A}_1$.\newline
The construction is simple: The tester $\mathcal{A}_3$ first runs $\mathcal{A}_1$ with input $\varphi$ and then runs $\mathcal{A}_2$ on each of the circuits produced by $\mathcal{A}_1.$ The high-level idea behind it is that we test an assignment to the input circuit $\varphi$ to $\mathcal{A}_1$ by testing that each of the `local views' of the various circuits produced by $\mathcal{A}_1$ is close to a satisfying assignment, where the latter process is carried out by the circuits produced by reduction $\mathcal{A}_2.$ This leads us to the following observation of Dinur and Reingold~\cite[Lemma 3.5]{DinurReingold}.\\

\begin{lemma}[Composing robust assignment-testers: Lemma 3.5 in \cite{DinurReingold}]
	\label{lem:DRcomposition}
	Suppose there exist robust assignment-testers $\mathcal{A}_1$ and $\mathcal{A}_2$ with parameters $(r_1,\ell_1,a_1,\rho_1)$ and  $(r_2,\ell_2,a_2,\rho_2)$ respectively. Then there also exists a robust assignment-tester $\mathcal{A}_3$ with parameters $(r_3,\ell_3,a_3,\rho_3)$
	where
	\begin{itemize}
		\item $r_3(n) = r_1(n) + r_2(a_1(n)),$
		\item $\ell_3(n) = \ell_2(a_1(n)),$
		\item $a_3(n) = a_2(a_1(n)),$
		\item $\rho_3(n) = \rho_1(n)\cdot\rho_2((a_1(n))).$
	\end{itemize}
\end{lemma}

\begin{proof}
    Following the high-level idea above, for a given Boolean circuit $\varphi$ of size $n$, the reduction $\mathcal{A}_3$ outputs a certain number of Boolean circuits as follows. The reduction $\mathcal{A}_3$ first runs $\mathcal{A}_1$ to produce Boolean circuits $\psi_1,\psi_2,\dots, \psi_{R_1}$, where $R_1=2^{r_1(n)}$, each of size at most $a_1(n)$. Now for each $\psi_i$, the reduction now runs $\mathcal{A}_2$ with input as $\psi_i$, to produce circuits $\psi_{i,j}$ for $j=1,2,\dots, R_2$, where $R_2=2^{r_2(a_1(n))}$ (since the circuit input to $\mathcal{A}_2$ is of size at most $a_1(n)$). The final output of $\mathcal{A}_3$ will be the circuits $(\psi_{i,j})_{i\in [R_1], j\in [R_2]}$ (hence, $r_3(n)=r_1(n)+r_2(a_1(n))$). We note that the running time of $\mathcal{A}_3$ is $\poly(R_1R_2,n)$, the size of the output circuits is at most $a_2(a_1(n))$, and each output circuit depends on at most $\ell_2(a_1(n))$ variables. 
    
    Let $\Sigma_1$ be the alphabet corresponding to auxiliary variables of $\mathcal{A}_1$ for input circuits of size $n$ and $\Sigma_2$ be the alphabet corresponding to the auxiliary variables of $\mathcal{A}_2$ for input circuits of size $a_1(n)$. Let $X$ denote the variables of $\varphi$, $Y$ denote the auxiliary variables of $\mathcal{A}_1$ for the circuit $\varphi$, and $Y_i$ denote the auxiliary variables of $\mathcal{A}_2$ for the circuit $\psi_i$ (for $i\in [R_1]$). Thus, the auxiliary variables of $\mathcal{A}_3$ are $Y \cup \cup_{i\in [R_1]} Y_i$ with alphabet $\Sigma_2$; i.e., we will treat the bits corresponding to $Y$ as being over $\Sigma_2$ as well (this can be without increasing the size and the number of variables each final circuit depends on). We will now show completeness and soundness.  
    
    If $\varphi(\sigma)=1$ for some $\sigma:X\to \{0,1\}$, by the completeness of $\mathcal{A}_1$, we note that there exists $\tau:Y\to \Sigma_1$ such that $\psi_i(\alpha_i)=1$ for all $i\in [R_1]$, where $\alpha_i := (\sigma \cup \tau)|_{\textrm{Vars}(\psi_i)}$. Now, by the completeness of $\mathcal{A}_2$, this means that there exists $\tau_i:Y_i \to \Sigma_2$ such that $\psi_{i,j}(\beta_{i,j})$ for all $j\in [R_2]$, where $\beta_{i,j}:=(\alpha_i \cap \tau_i)|_{\mathrm{Vars}}(\psi_{i,j})$. Therefore, we see that for the assignment $\sigma \cup \tau \cup_{i\in [R_1]}\tau_i$, all the circuits $\psi_{i,j}$ are satisfied. 

    To show soundness, let $\delta:=\delta(\sigma, \varphi^{-1}(1))$ and $\tau:Y\to \Sigma_1$ and $\tau_{i}:Y_i\to \Sigma_2$ be arbitrary assignments to the auxiliary variables of $\mathcal{A}_1$ and $\mathcal{A}_2$, for $i\in [R_1]$. Again, let $\alpha_i := (\sigma \cup \tau)|_{\textrm{Vars}(\psi_i)}$ and $\beta_{i,j}:=(\alpha_i \cap \tau_i)|_{\mathrm{Vars}}(\psi_{i,j})$. By the robust soundness condition of $\mathcal{A}_1$, we have that
    $$\E_{i\in [R_1]}[\delta(\alpha_i, \psi_i^{-1}(1))] \ge \rho_1 \delta .$$ Similarly, for each $i\in [R_1]$, by the robust soundness of $\mathcal{A}_2$, we have that 
    $$\E_{j\in [R_2]}[\delta(\beta_{i,j}, \psi_{i,j}^{-1}(1))] \ge \rho_2 \delta(\alpha_i, \psi_i^{-1}(1)),$$ where $\rho_2=\rho_2(a_1(n))$ since the size of $\psi_i$ is at most $a_1(n)$.
    Taking expectation over a uniform random $i\in [R_1]$ on both sides of the above equation, we get the desired robust soundness for $\mathcal{A}_3$:
    $$\E_{i\in [R_1], j\in [R_2]} [\delta(\beta_{i,j}, \psi_{i,j}^{-1}(1))] \ge \rho_2 \E_{i\in [R_1]}[\delta(\alpha_i, \psi_{i}^{-1}(1))] \ge \rho_1 \rho_2 \delta.$$ This finishes the proof of the composition lemma. 
\end{proof}

\paragraph{Turning our PCPs into robust assignment-testers.} To use the above lemma, we instantiate our main theorem with the two different choices of varieties used in \Cref{lem:polylogn-pcp} and \Cref{lem:neps-pcp}. While these statements only yield PCPs (a weaker object than assignment-testers) for the $3$-$\mathsf{COLOR}$ problem, we show that they easily yield assignment-testers for Boolean circuits using two basic ingredients: basic properties of standard reductions showing $3$-$\mathsf{COLOR}$ is NP-hard and the local correctability of the polynomial witnesses in the PCP proof of \Cref{thm:main-thm} using \Cref{thm:local-correction}.

The following is a basic property of standard reductions form Circuit-SAT to $3$-$\mathsf{COLOR}$. One can prove it e.g. by using the standard Tseitin transformation (reducing Circuit-SAT to $3$-SAT) followed by the reduction from $3$-SAT to $3$-$\mathsf{COLOR}$ in \cite[Theorem 2.1]{GareyJohnsonStockmeyer}.\\

\begin{lemma}
	\label{lem:3color-redn}
	There is a polynomial-time reduction from Circuit-SAT to 3-coloring that satisfies the following. On input a Boolean circuit $\varphi$, the graph $G = (V,E)$ produced by the reduction has the property that for any satisfying assignment $\sigma:X\rightarrow \{0,1\}$ of $\varphi$, there is a proper $3$-coloring $\chi:V\rightarrow \{-1,0,1\}$ such that $|V| = O(|\varphi|)$, $\chi$ restricts to $\sigma$ on a fixed subset $V_0$ of $V$ (here $V_0$ is in $1$-$1$ correspondence with $X$) and $\chi(v_0) = -1$ for some fixed $v_0\in V\setminus V_0$. Furthermore, \emph{any} proper $3$-coloring $\chi:V\rightarrow \{-1,0,1\}$ satisfying $\chi(v_0) = -1$, upon restriction to $V_0$, yields a satisfying assignment $\sigma':X\rightarrow \{0,1\}$ of the circuit $\varphi.$
\end{lemma}

Given the above, we can easily modify the PCPs from the proofs of \Cref{thm:main-thm} to yield assignment-testers for CircuitSAT. We will show how to do this below.

Moreover, we will use an idea of Dinur and Reingold to make these assignment-testers \emph{robust} using the simple process of encoding the input symbols by an explicit asymptotically good error-correcting code. More precisely, we use the following lemma.\\

\begin{lemma}[Robustization: Lemma 3.6 in~\cite{DinurReingold}]
	\label{lem:robustization}
	There is an absolute constant $c_1$ such that if there is an assignment-tester $\mathcal{A}$ with parameters $(r,\ell,a,\rho)$, then there is also a \emph{robust} assignment-tester $\mathcal{A}'$ with parameters $(r' = r,\ell' = c_1 \ell a, a'=c_1 \ell a, \rho'= \rho/(c_1 \ell))$.
\end{lemma}

We are now ready to state the main lemma of this section, which says that the PCPs from the previous section can be turned into robust assignment-testers with suitable parameters.\\

\begin{lemma}
	\label{lem:PCPstoRATs}
	There exist absolute constants $c,\rho>0$ such that the following holds. Assume that for every $n\geq 1$, there exist $q = q(n),m = m(n),d = d(n),k = k(n)$ such that $q \geq c d^3$ and a
	variety $V_n\subseteq \F_q^m$ of size $\omega(n)$ constructible in time $\poly(n)$ with extension degree $d$ and Macaulay complexity $k$. Then, there is a robust assignment-tester with parameters
	\[\left( r={O((k+m)\log q)},\ell = O(d\log q),a=O(d\log q),\rho\right)\]  
\end{lemma}

\begin{proof}
	The constant $c$ is chosen from \Cref{thm:main-thm}. We can assume without loss of generality that $c \geq 100.$

	We can now describe the reduction $\mc{A}$ behind the assignment-tester (we will make it robust below). The reduction first reduces the given instance $\varphi$ of CircuitSAT to an instance $G = (V,E)$ of $3$-$\mathsf{COLOR}$ using the polynomial-time reduction described in \Cref{lem:3color-redn} (so $|V| = O(n)$). Fix $V_0\subseteq V$ (which is $1$-$1$ correspondence with the set of variables $X$ of $\varphi$) and $v_0\in V\setminus V_0$ as mentioned above.

	The reduction $\mathcal{A}$ now constructs a variety $V_n$ as in the hypothesis of the lemma. Note that $|V_n| = \omega(n) \geq |V|$. After adding some isolated vertices to $G$ if required, we can identify $V$ with $V_n$.

	We now consider the assignment-tester specified in \Cref{algo:AT} with $c_2$ being an absolute constant that we will choose below. For any choice of this constant, the algorithm uses $r_1 = O((k+m)\log q)$ many random bits, makes $\ell_1 = O(1)$ queries to its input oracles $\sigma$ and $\tau$ where the former is defined over the Boolean alphabet and the latter is defined over an alphabet of size $O(d\log q).$ The reduction $\mc{A}$ iterates over all sequences $b$ of $r_1$ random coin tosses used by the algorithm and for each it produces a Boolean circuit $\psi_b$ that performs the checks specified in the algorithm.

	\begin{algobox}
		\begin{algorithm}[H]
			\caption{Assignment-Tester}
			\label{algo:AT}

			\DontPrintSemicolon

			\KwIn{Degree parameter $d$ and oracle access to $(\sigma,\tau)$ where $\sigma:X\rightarrow \{0,1\}$, $\tau:Y\rightarrow \Sigma$, where $\Sigma$ is the alphabet of the PCP verifier from \Cref{thm:main-thm} and $|Y|=q^{O(k+m)}$ is the length of the proof.}

			\vspace{3mm}
			Run the PCP verifier from \Cref{thm:main-thm} on $\tau$ independently $c_2$ times. If the PCP verifier rejects even once, then \Return{\texttt{Reject}}
			\tcp*{Random bits $O(c_2(k+m)\log q)$, number of queries $O(1)$, alphabet size $O(d\log q)$.}

			\vspace{3mm}

			Run $\mathcal{LC}_d^{\hat{\chi},\hat{\chi}'}(v_0)$ and if the algorithm rejects or returns anything other than $-1$, then \Return{\texttt{Reject}}
			\tcp*{Random bits $O(m\log q)$, number of queries $O(1)$.}
			\vspace{2mm}

			Sample $v\in V_0$ uniformly at random and run $\mathcal{LC}_d^{\hat{\chi},\hat{\chi}'}(v)$. If the algorithm rejects or returns anything other than $\sigma(v)$, then \Return{\texttt{Reject}}
			\tcp*{Random bits $O(m\log q)$, number of queries $O(1)$.}
			\vspace{2mm}

			\Return{\texttt{Accept}}
		\end{algorithm}
	\end{algobox}

	The number of circuits produced is $2^{r_1} = q^{O(k+m)}$, each of which queries $\ell_1 = O(1)$ locations in the string $\sigma \cup \tau.$ The computations performed by $\psi_b$ are all efficiently computable functions of the queried bits, and hence the circuit $\psi_b$ has size that is polynomial in the number of bits queried, which is $O(d\log q).$ The alphabet $\Sigma$ is the same as that of the PCP verifier and hence has cardinality $q^{O(d)}.$

	To show that this is an assignment-tester, it suffices to argue the completeness and soundness criteria. Completeness is trivial from the definition of the PCP and the properties of the NP-completeness reduction argued above.

	For soundness, let $\sigma:X\rightarrow \{0,1\}$ be an arbitrary assignment and let $\delta:=\delta(\sigma,\varphi^{-1}(1))$. We need to show that the probability that \Cref{algo:AT} rejects is at least $3\delta/4$, since this also means that at least $3\delta/4$ of the circuits $\psi_b$ reject.

	By the soundness of the PCP verifier \Cref{claim:robust-soundness-pcp} we know that it rejects with constant probability unless $\hat{\chi}$ is at distance at most $\eta = 0.01$ from a degree $d$ polynomial $P:\F_q^m\rightarrow \F_q$ that is a low-degree extension of a proper $3$-coloring of $G.$ Since we repeat the tests of the PCP verifier $c_2$ many times, the acceptance probability in case this does not hold is $\exp(-\Omega(c_2))\leq 1/4$ as long as $c_2$ is large enough.

	So from now on, we assume that $\hat{\chi}$ is $\eta$-close to a degree $d$ polynomial $P$ which is a low-degree extension of a proper $3$-coloring $\chi:V\rightarrow \{-1,0,1\}.$ If $\chi(v_0) \neq -1,$ then by \Cref{thm:local-correction}, the next step rejects with probability at least
	$$1-2\sqrt{\eta}-\frac{d}{q-1}\geq 1-2\sqrt{\eta}-\frac{1}{c}\geq \frac{3}{4}.$$
	Thus, we can assume that $\chi(v_0) = -1.$

	Finally, since $\chi$ is a proper $3$-coloring of $G$ and $\chi(v_0) = -1$, it follows that the restriction of $\chi$ to $V_0$ defines a satisfying assignment of $\varphi$. Since $\sigma$ is $\delta$-far from any satisfying assignment of $\varphi$, it follows that for a random $v\in V_0$ chosen in the next step, the probability that $\sigma$ and $\chi$ differ at $v$ is at least $\delta$. Further, by \Cref{thm:local-correction}, the probability that $\mathcal{LC}_d^{\hat{\chi},\hat{\chi}'}(v)$ returns $\chi(v)$ or rejects is at least $3/4$ (as in the previous paragraph). Thus, the chance that this step rejects is at least
	$3\delta/4 $. This concludes the proof of the soundness of $\mathcal{A}.$

	We have thus shown that $\mathcal{A}$ is an assignment-tester with parameters
	$$(O(k+m)\log q, \ell_1=O(1), O(d\log q),3/4).$$
	In order to make the assignment-tester robust, we simply apply the Robustization lemma~\Cref{lem:robustization}. This yields a {\em robust} assignment-tester with parameters $$(O(k+m)\log q, O(d\log q), O(d\log q),3/(4c_1\ell_1))$$ for some constant $c_1>0$, and concludes the proof of \Cref{lem:PCPstoRATs}.
\end{proof}

\paragraph{Proving the PCP theorem.} To conclude the proof of the PCP theorem, we apply the above lemma to the PCPs given by specific instantiations of $V$ already seen above. \\

\begin{lemma}[Outer robust assignment-tester]
	\label{lem:polylogn-rat}
	There exists a robust assignment-tester with parameters $$(r_1=O(\log n),\ell_1 = O(\log n)^2, a_1=O(\log n)^2, \rho_1 = \Omega(1)).$$
\end{lemma}

\begin{proof}
	Set $q = (\log n)^6$ a power of $3,$ $V = H^m$ for $H\subseteq \F_q$ of size $\log n$, and $m = O(\log n/\log \log n)$ such that $|V|=\poly(n)$. By~\Cref{cor:Hhamm-wt-c}, we thus have that the extension degree of $V$ is $d=O((\log n)^2/\log \log n)$ and Macaulay complexity is $k=O(\log n/\log\log n)$. The lemma then follows by applying \Cref{lem:PCPstoRATs}.
\end{proof}

\begin{lemma}[Inner robust assignment-tester]
	\label{lem:neps-rat}
	For every $\alpha>0$, there exists a robust assignment-tester with parameters
	$$(r_2 = O(n^\alpha),\ell_2 = O(1), a_2 = O(1) ,\rho_2 = \Omega(1)).$$
\end{lemma}

\begin{proof}
	Set $V = (\{0,1\}^{m/c}_{\leq 1})^c\subseteq \F_q^m$ for suitably large constant integers $c \ge 2/\alpha$ (i.e., we are treating $\alpha$ and $c$ as constants) and $q$ which a power of 3, and $m = O(n^{1/c})$ such that $|V|=\poly(n)$. By~\Cref{cor:hamm-wt-c}, we thus have that the extension degree of $V$ is $d\le c$ and Macaulay complexity is $k=O(m^2)$. Taking $q\ge cd^3$ to be a sufficiently large constant, by~\Cref{lem:PCPstoRATs}, we obtain a robust assignment-tester with parameters $(O(n^\alpha), O(1),O(1),\Omega(1))$
\end{proof}

We can now prove the PCP theorem.

\begin{proof}[Proof of \Cref{thm:pcp-thm}]
	We compose the $(r_1,\ell_1,a_1,\rho_1)$ robust assignment-testers from~\Cref{lem:polylogn-rat} with the $(r_2,\ell_2,a_2,\rho_2)$ robust assignment-tester from~\Cref{lem:neps-rat} corresponding to $\alpha=1/2$, by using~\Cref{lem:DRcomposition}. This leads to a robust assignment-tester with parameters $(r_3, \ell_3,a_3,\rho_3)$, where $r_3(n)=r_1(n) + r_2(a_1(n)) = O(\log n)$, $\ell_3=\ell_2(a_1(n))=O(1), a_3=a_2(a_1(n))=O(1)$, and $\rho_3=\rho_1\rho_2 =\Omega(1)$. Combined with \Cref{obs:testersPCPs}, this implies the PCP theorem.
\end{proof}

\medskip

\section*{Acknowledgments}

We are thankful to Yassine Ghannane for pointing us to the relevant definitions of Macaulay bases. We are also thankful to the anonymous STOC reviewers for their valuable comments and bringing the work of Khot~\cite{khot2006ruling} to our attention. 

\printbibliography[
	heading=bibintoc,
	title={References}
] %Prints the entire bibliography with the title "References"

\appendix
\section{Relation between Gröbner Bases and Macaulay Bases}\label{GrobnerAppendix}
In \Cref{sec:generatingset}
we introduced Macaulay bases.
We here show that all Gr\"obner bases are Macaulay bases,
and that Macaulay bases are also well-behaved under Cartesian products.\\

\begin{definition}
	An ordering of monomials in \(k[\x]\) is called admissible,
	if every monomials \(M, N, L \in k[\x]\) satisfies
	\begin{enumerate}
		\item \(M \leq N \) implies \(ML \leq NL\).
		\item \(M \leq ML\).
	\end{enumerate}
	if an admissible ordering further satisfies
	\begin{enumerate}
		\setcounter{enumi}{2}
		\item \(\deg(M) < \deg(N)\) implies \(M < N\)
	\end{enumerate}
	we call it a graded ordering.
	For a polynomial \(P\), we define \(\LM(P)\)
	as the monomial in \(P\) of maximal order.
\end{definition}

\begin{example}
	For the lexicographic ordering
	we have \(M < N\)
	if there exist \(i\) such that the exponent of \(x_{j}\) in \(M\) is
	equal to the exponent of \(x_{j}\) in \(N\) for \(j <i\)
	and the exponent of \(x_{i}\) in \(M\) is strictly smaller than the
	exponent of \(x_{i}\) in
	\(N\).

	For the graded lexicographic,
	we have \(M < N\) if \(\deg(M) < \deg(N)\) or if
	\(\deg(M)=\deg(N)\) and \(M < N\) in the lexicographic ordering.
\end{example}

\begin{definition}
	fix an admissible ordering of monomials in \(k[\x]\).
	A generating set \(G\) of an ideal
	\(I \subseteq k[\x]\)
	is a Gröbner basis with respect to that ordering,
	if the leading monomial of every polynomial in \(I\)
	is a multiple of a leading monomial
	of a polynomial in \(G\).

	If the ordering is graded, we say that \(G\) is a graded Gröbner basis.
\end{definition}

\begin{lemma}
	Let \(\Grobner\) be a Gröbner basis for a graded ordering.
	Then \(\Grobner\) is a Macaulay basis.
\end{lemma}

\begin{proof}
	If \(\Grobner\) is a Gröbner basis for \(\I\),
	then a polynomial \(f\) is in \(\I\)
	if every complete lead reduction of \(f\) results in the zero polynomial.
	Every step of the reduction will be of the form
	\begin{align*}
		f' - \frac{\LM(f')}{\LM(g)}g
	\end{align*}
	for some \(g \in \Grobner\) and \(f'\) being an intermediate result in the reduction.

	For a graded ordering we have by definition \(\deg(\LM(f))=\deg(f)\) for every \(f\), so
	\begin{align*}
		\deg \left(\frac{\LM(f')}{\LM(g)}g
		\right)
		=
		\deg(f')
		\leq
		\deg(f)
	\end{align*}
\end{proof}

\begin{remark}\label{rem:mingrob}
	If \(\Grobner\) is a Gr\"obner basis of smallest
	size,
	then is not necessarily a Macaulay basis of smallest size.
	For example, the two polynomials spanning the ideal
	\begin{align*}
		\left( x_{1}^{2}, x_{1}x_{2}-x_{2}^{2}\right)
	\end{align*}
	form a Macaulay basis, as they both are
	homogeneous of degree 2.
	However, we also have
	\begin{align*}
		x_{2}^{3}= x_{2} \cdot  x_{1}^{2} -
		\left(x_{1}+x_{2}\right) \cdot \left( x_{1}x_{2}-x_{2}^{2}\right)
	\end{align*}
	so \(x_{2}^{3} \in \left(x_{1}^{2}, x_{1}x_{2}-x_{2}^{2}\right)\),
	and so in the graded lexicographic ordering
	\(x_{1}^{2}, x_{1}x_{2}, x_{2}^{3}\) are all leading monomials.
	Since no monomials of degree 2 can divide two of these monomials,
	a Gr\"obner basis must contain at least 3 elements.
\end{remark}

\addtocontents{toc}{\protect\setcounter{tocdepth}{1}}

\end{document}